\documentclass[aps,PRA,reprint,superscriptaddress]{revtex4-2}
\usepackage{graphicx}
\usepackage{braket}
\usepackage{amsmath}
\usepackage{amsthm}
\usepackage{amssymb}
\usepackage{dsfont}
\usepackage[colorlinks,
            linkcolor=blue,
            anchorcolor=black,
            citecolor=blue,
			urlcolor=blue]{hyperref}
			
\usepackage[ruled]{algorithm2e}
\usepackage{makecell}
\usepackage{array}
\usepackage{color}
\usepackage{multirow}

\newtheorem{proposition}{Proposition}
\newtheorem{definition}{Definition}
\newtheorem{corollary}{Corollary}

\begin{document}


\title{The four-dimensional Chamon code}


\author{Zhipeng Liang}
\affiliation{Harbin Institute of Technology, Shenzhen. Shenzhen, 518055, China}
\author{Xuan Wang}
\email[]{wangxuan@cs.hitsz.edu.cn}
\affiliation{Harbin Institute of Technology, Shenzhen. Shenzhen, 518055, China}


\date{\today}

\begin{abstract}
Fracton models have attracted considerable interest as candidates for quantum memories because of their unconventional ground-state degeneracy (GSD) and restricted-mobility excitations. The four-dimensional (4D) Chamon code introduced in our previous work \cite{liang2025high} is constructed via the 4D XYZ product of two two-dimensional (2D) toric codes. Its GSD grows exponentially with the system size, similar to that of the three-dimensional (3D) Chamon code, suggesting that it may be regarded as a 4D generalization of the 3D Chamon code. However, the excitation properties of the 4D Chamon code have not been studied in depth, and a high-performance decoding strategy is still lacking. In this work, we first establish the correspondence between the algebraic structure of the 4D Chamon code and the 4D lattice, thereby characterizing the geometric distributions of qubits and stabilizers. Second, we show that the 4D Chamon code supports three types of restricted-mobility excitations analogous to those of the 3D Chamon code, further supporting its interpretation as a 4D generalization of the 3D Chamon code. Finally, we uncover two structural properties relevant to decoding: a hyperplane symmetry and a projection-induced 2D toric-code structure. By exploiting these properties, we develop a two-layer decoding strategy that decomposes the original decoding problem into multiple independent and parallelizable subproblems. Numerical simulations show that the proposed decoder substantially outperforms BP-OSD in decoding accuracy, demonstrating the benefit of incorporating the intrinsic geometric and algebraic structures of the code into decoder design.
\end{abstract}


\maketitle

\section{Introduction}
\label{1}
Fracton models \cite{chamon2005quantum,vijay2015new,vijay2016fracton,haah2011local,nandkishore2019fractons,pretko2020fracton,ma2017fracton,tantivasadakarn2021non,shen2022fracton} constitute a class of quantum many-body systems that have attracted considerable attention in recent years. Unlike conventional topologically ordered phases, fracton models are characterized by elementary excitations with restricted mobility and GSD that grows exponentially with the system size. Representative fracton models include the 3D Chamon model \cite{chamon2005quantum,bravyi2011topological}, Haah's code \cite{haah2011local}, and the X-cube model \cite{vijay2016fracton,shen2022fracton}.

In our previous work \cite{liang2025high}, we proposed the 4D XYZ product code construction, which allows one to construct a new quantum stabilizer code \cite{gottesman1997stabilizer} from two 2D toric codes \cite{kitaev2003fault}. Further analysis shows that the GSD of this code grows exponentially with the system size, which exhibits a mathematical structure similar to that of the 3D Chamon code. Based on this similarity, we regard this code as a 4D generalization of the 3D Chamon code and refer to it as the 4D Chamon code. However, as a representative 3D fracton model, the 3D Chamon code is characterized not only by system-size-dependent GSD, but also by excitations with restricted mobility. Therefore, to further demonstrate that the 4D Chamon code introduced in Ref. \cite{liang2025high} is indeed a 4D generalization of the 3D Chamon code, it is necessary to show that it also supports restricted-mobility excitations similar to those in the 3D Chamon code. In addition, the 4D Chamon code currently lacks a high-accuracy decoding algorithm that can exploit its geometric and algebraic structures. Therefore, the motivation of this paper is twofold: first, to determine whether the 4D Chamon code can be understood as a 4D generalization of the 3D Chamon code; and second, to design a high-performance decoding algorithm based on its geometric and algebraic structures.

The main contributions of this paper are threefold. First, we establish the correspondence between the algebraic structure of the 4D Chamon code and the 4D periodic lattice, and explicitly identify the geometric distributions of physical qubits and stabilizers on the lattice. Second, we analyze the excitation properties of the 4D Chamon code. The results show that the code possesses a key feature of fracton models, namely, excitations with restricted mobility, and that its excitation types are similar to those of the 3D Chamon code. These results demonstrate that the 4D Chamon code is not only a new 4D fracton model, but can also be regarded as a 4D generalization of the 3D Chamon code. Finally, we reveal that the 4D Chamon code possesses symmetry properties and projection-induced 2D toric-code  structures. By exploiting these structures, we propose a two-layer decoding algorithm tailored to the 4D Chamon code. Numerical simulation results show that the proposed decoder effectively exploits the symmetry and projection structures of the code and achieves substantially higher decoding accuracy than the BP-OSD decoder \cite{panteleev2021degenerate}.

This paper is organized as follows. Sec.~\ref{2} introduces the necessary preliminaries, including geometric notations, chain complexes, and the CSS variant of the 4D XYZ product construction. Sec.~\ref{The 4D Chamon Code: geometric structure and excitations} establishes the correspondence between the algebraic structure of the 4D Chamon code and the 4D periodic lattice, thereby determining its geometric structure and excitation types. Sec.~\ref{two-layer decoding algorithm} presents the two-layer decoding algorithm for the 4D Chamon code, which is developed by exploiting its symmetry and projection structures. Sec.~\ref{simulation result} evaluates the performance of the proposed two-layer decoding algorithm, and Sec.~\ref{conclusion} concludes the paper.

\section {Preliminaries}
\label{2}
\subsection{Geometric notations}
\label{Geometric notations}
To describe the geometric objects in a $D$-dimensional lattice, we adopt the coordinate representation introduced in Ref. \cite{li2020fracton}, in which a $d$-dimensional geometric object ($0\leq d\leq D$) is specified by the coordinate of its geometric center. Consider a $D$-dimensional periodic lattice $\Lambda_D$ with lattice constant 2. Under this representation, the coordinate of the geometric center of any $d$-dimensional object contains exactly $d$ odd entries and $D-d$ even entries.

Accordingly, any geometric object in the 4D lattice can be classified according to the parity pattern of the coordinate of its geometric center, with its dimension given by the number of odd entries.
For example, $(0,0,0,0)$ represents a vertex; $(1,0,0,0)$ represents an edge centered at $(1,0,0,0)$; $(1,1,0,0)$ represents a face centered at $(1,1,0,0)$; $(1,1,1,0)$ represents a 3D cube centered at $(1,1,1,0)$; and $(1,1,1,1)$ represents a 4D hypercube centered at $(1,1,1,1)$.

Furthermore, the type of a geometric object is determined by the positions of the odd entries in its geometric-center coordinate. For an edge, if the unique odd entry is in the $\omega$-coordinate, with $\omega\in\{x,y,z,w\}$, then the edge is oriented along the $\omega$-axis. For example, $(1,0,0,0)$ represents an $x$-oriented edge. For a 2D face, if the two odd entries are in the $\omega_1$- and $\omega_2$-coordinates, then the face lies in the $\omega_1\omega_2$-plane. For example, $(1,1,0,0)$ represents an $xy$-face. For a 3D cube, if the unique even entry is in the $\omega$-axis, then the cube spans the other three coordinate directions and is referred to as a cube perpendicular to the $\omega$-axis. For example, $(1,1,1,0)$ represents a cube perpendicular to the $w$-axis.

Accordingly, a 4D lattice contains four types of edges, oriented along the $x$-, $y$-, $z$-, and $w$-axes, respectively; six types of faces, lying in the $xy$-, $xz$-, $xw$-, $yz$-, $yw$-, and $zw$-planes, respectively; and four types of cubes, perpendicular to the $x$-, $y$-, $z$-, and $w$-axes, respectively.

\subsection {Chain complex}
\label{chain complex}
This section introduces the concept of chain complex, which will help to better understand the CSS variant of the 4D XYZ product introduced in Sect. \ref{CSS 4D XYZ product}.

A chain complex $\mathfrak{C}$ with length $L$ is a collection of $L+1$ vector spaces $C_0,\ C_1,\cdots,C_L$ and $L$ linear maps (which are also called boundary operators) $\partial_i:C_i\rightarrow C_{i+1}\ \left(0\le i\le L-1\right)$, namely,
\begin{equation}
	\mathfrak{C}=C_0\stackrel{\partial_0}{\longrightarrow}C_1\stackrel{\partial_1}{\longrightarrow}\cdots\stackrel{\partial_{L-1}}{\longrightarrow}C_L
\end{equation}
which satisfies $\partial_i\partial_{i-1}=0$ for all $0\le i\le L-1$.

If we consider vector space over $\mathbb{F}_2$, namely, $C_i:=F_2^{n_i}$, a chain complex $\mathfrak{C}$ with length 2 naturally corresponds to a CSS code $C\left(\mathfrak{C}\right)$, namely,
\begin{equation}
	C\left(\mathfrak{C}\right)=\mathbb{F}_2^{m_z}\stackrel{H_z^T}{\longrightarrow}\mathbb{F}_2^N\stackrel{H_x}{\longrightarrow}\mathbb{F}_2^{m_x}
\end{equation}
where the commutation condition $H_xH_z^T=\mathbf{0}$ is naturally satisfied.

\subsection{The CSS variant of the 4D XYZ product}
\label{CSS 4D XYZ product}
The 4D XYZ product proposed in Ref. \cite{liang2025high} makes use of two CSS codes to construct a non-CSS code, which can be transformed into a CSS code by finite-depth unitary circuits, please see Appendix A in Ref. \cite{liang2025high} for more details. 
 
In this paper, we consider this CSS variant the 4D XYZ product. Formally, Giving two length-2 chain complexes $\mathfrak{C}_1=C_{-1}\stackrel{H_{z_1}^T}{\longrightarrow}C_0\stackrel{H_{x_1}}{\longrightarrow}C_1$ and $\mathfrak{C}_2={\widetilde{C}}_{-1}\stackrel{H_{z_2}^T}{\longrightarrow}{\widetilde{C}}_0\stackrel{H_{x_2}}{\longrightarrow}{\widetilde{C}}_1$, which corresponds to two CSS codes $C\left(\mathfrak{C}_1\right)$ and $C\left(\mathfrak{C}_2\right)$, one can construct the tensor-product structure as shown in Fig. \ref{4DXYZproduct},
where the qubits are divided into $A$, $B$, $C$, $D$, $E$ five parts and the stabilizer generators are divided into $S$, $T$, $U$, $V$ four parts. The corresponding stabilizer matrix $\mathcal{S}$ is
\begin{widetext}
	\begin{equation}
		\label{4DXYZ stabilizer_CSS}
		\mathcal{S} =\begin{bmatrix}
			S\\
			T\\
			U\\
			V
		\end{bmatrix} = \begin{bmatrix}
			Z^{\left(I_{m_1}\otimes H_{x_2}^T\right)} &Z^{\left(I_{m_1}\otimes H_{z_2}^T\right)} &Z^{\left(H_{z_1}\otimes I_{n_B}\right)} &I &I\\
			X^{\left(H_{z_1}^T\otimes I_{m_4}\right)} &I &X^{\left(I_{n_A}\otimes H_{x_2}\right)} &X^{\left(H_{x_1}^T\otimes I_{m_4}\right)} &I\\
			I &X^{\left(H_{z_1}^T\otimes I_{m_3}\right)} &X^{\left(I_{n_A}\otimes H_{z_2}\right)} &I &X^{\left(H_{x_1}^T\otimes I_{m_3}\right)}\\
			I &I &Z^{\left(H_{x_1}\otimes I_{n_B}\right)} &Z^{\left(I_{m_2}\otimes H_{x_2}^T \right)} &Z^{\left(I_{m_2}\otimes H_{z_2}^T\right)}
		\end{bmatrix}
	\end{equation}
\end{widetext}
where $m_1$, $m_2$, $m_3$, $m_4$, $n_A$, $n_B$ are the dimensions of vector spaces $C_{-1}$, $C_1$, ${\widetilde{C}}_{-1}$, ${\widetilde{C}}_1$, $C_0$ and ${\widetilde{C}}_0$, respectively. The total number of qubits is
\begin{equation}
	\label{4DXYZ code length}
	N=m_1m_4+m_1m_3+n_An_B+m_2m_4+m_2m_3
\end{equation}
\begin{figure}[htbp]
	\centering
	\includegraphics[width=0.48\textwidth]{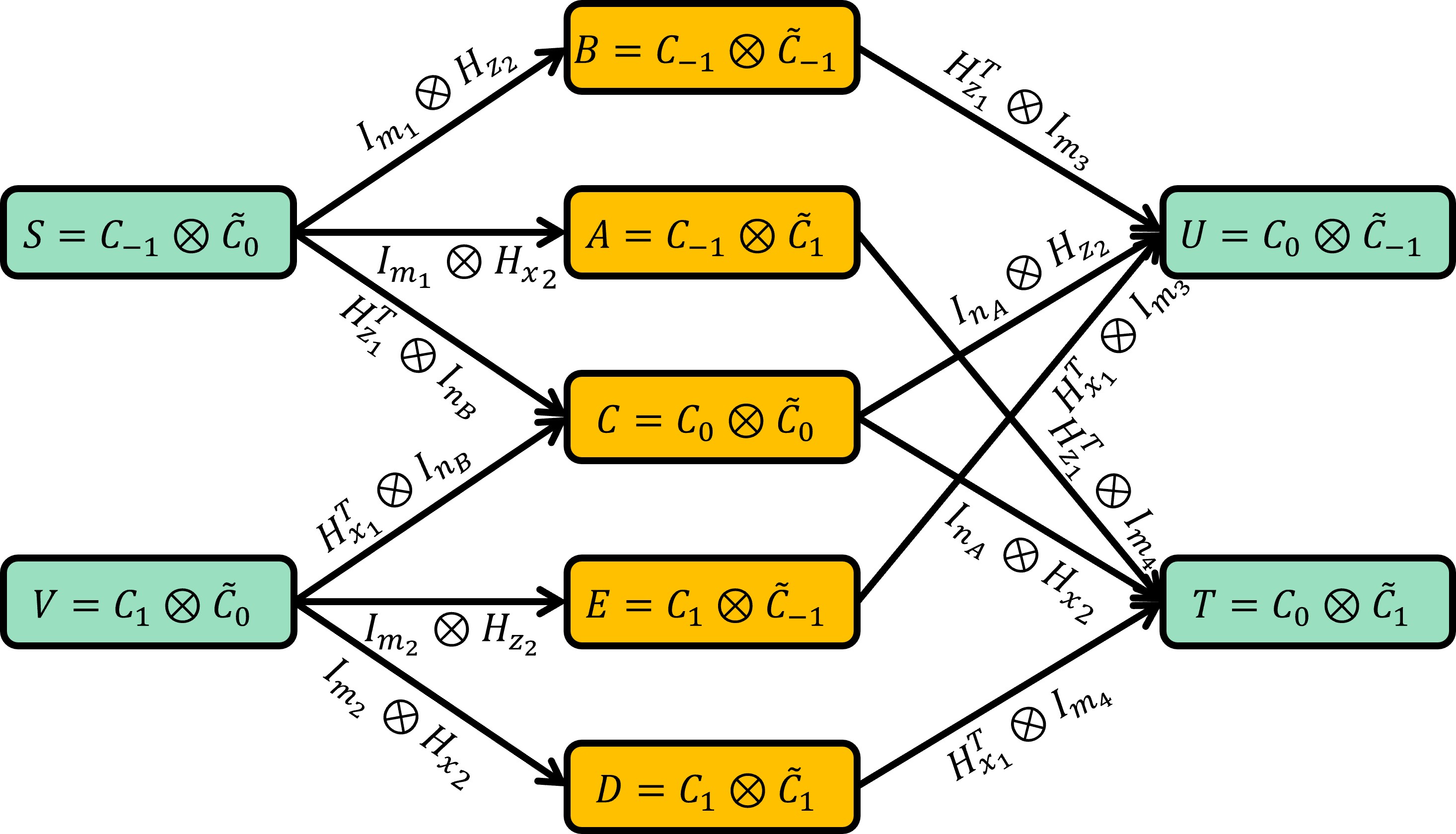}
	\caption{The tensor-product structure of the CSS variant of the 4D XYZ product of two length-2 chain complexes $\mathfrak{C}_1=C_{-1}\stackrel{H_{z_1}}{\longrightarrow}C_0\stackrel{H_{x_1}}{\longrightarrow}C_1$ and $\mathfrak{C}_2={\widetilde{C}}_{-1}\stackrel{H_{z_2}}{\longrightarrow}{\widetilde{C}}_0\stackrel{H_{x_2}}{\longrightarrow}{\widetilde{C}}_1$.}
	\label{4DXYZproduct}
\end{figure}

\section{The 4D Chamon Code: geometric structure and excitations}
\label{The 4D Chamon Code: geometric structure and excitations}
The 4D Chamon code developed in Ref. \cite{liang2025high} is a non-CSS code constructed by the 4D XYZ product of two 2D toric codes. As described Appendix A in Ref. \cite{liang2025high}, the non-CSS code constructed by the 4D XYZ product can be transformed into a CSS code through a finite-depth unitary circuit. Therefore, applying the CSS variant of the 4D XYZ product to two 2D toric codes yields a CSS variant of the 4D Chamon code. For convenience of presentation, throughout the remainder of this paper, unless otherwise specified, the term "4D Chamon code" refers to this CSS variant.

First, this section establishes the correspondence between the algebraic structure of the 4D Chamon code and the 4D periodic lattice, and explicitly identify the geometric distributions of physical qubits and stabilizers on the lattice. Second, we analyze the excitation properties of the 4D Chamon code and show that the code possesses a key feature of fracton models, namely, excitations with restricted mobility, and that its excitation types are similar to those of the 3D Chamon code.

\subsection{Geometric structure}
\label{4D Chamon code geometry structure}
The 4D Chamon code is constructed by the CSS variant of the 4D XYZ product of two 2D toric codes. Let $C_1$ denote the first 2D toric code, defined on a $2L_x\times 2L_y$ 2D periodic lattice with $X$-type and $Z$-type parity-check matrices $H_{z_1}$ and $H_{x_1}$, respectively. Similarly, Let $C_2$ denote the second 2D toric code, defined on a $2L_z\times 2L_w$ 2D periodic lattice with $X$-type and $Z$-type parity-check matrices $H_{z_2}$ and $H_{x_2}$, respectively.

For a 2D toric code, the $Z$-type and $X$-type stabilizer generators are associated with faces and vertices, respectively, whereas the physical qubits are located on edges. According to the geometric representation introduced in Sect. \ref{Geometric notations}, the coordinate of a face center contain two odd entries, namely, $(odd,odd)$, while that of a vertex contain two even entries, namely, $(even,even)$. The coordinate of an edge center contain one odd entry and one even entry, and therefore take the form $(odd,even)$ or $(even,odd)$. Consequently, after applying the 4D XYZ product to two 2D toric codes, the geometric types of the resulting stabilizers and physical qubits can be identified directly from the parity patterns of their 4D coordinates.

We first determine the geometric types of the physical qubits. As illustrated by the tensor-product structure in Fig.~\ref{4DXYZproduct}, the physical qubits of the 4D Chamon code are associated with five vector spaces, denoted by $A$, $B$, $C$, $D$, and $E$.

In particular, for $A=C_{-1}\otimes\widetilde{C}_{1}$, a basis for $A$, denoted by $\mathcal{B}(A)$, can be written as
\begin{equation}
	\mathcal{B}(A)=\{b_i\otimes c_j \mid i=1,\cdots,m_1,\ j=1,\cdots,m_4\}.
\end{equation}
Here, $\mathcal{B}(C_{-1})=\{b_1,\cdots,b_{m_1}\}$ and $\mathcal{B}(\widetilde{C}_1)=\{c_1,\cdots,c_{m_4}\}$ are bases for the vector spaces $C_{-1}$ and $\widetilde{C}_{1}$, respectively. According to the correspondence between the chain-complex structure and the lattice geometry, each element of $\mathcal{B}(C_{-1})$ corresponds to a $Z$-type stabilizer generator of $C_1$, whereas each element of $\mathcal{B}(\widetilde{C}_1)$ corresponds to an $X$-type stabilizer generator of $C_2$. Therefore, each $b_i\in\mathcal{B}(C_{-1})$ corresponds to a face whose geometric-center coordinate has the parity pattern $(odd,odd)$, while each $c_j\in\mathcal{B}(\widetilde{C}_1)$ corresponds to a vertex whose geometric-center coordinate has the parity pattern $(even,even)$. Consequently, every tensor-product basis element $b_i\otimes c_j\in\mathcal{B}(A)$ corresponds to a 4D geometric object whose geometric-center coordinate has the parity pattern $(odd,odd,even,even)$, namely, a face lying in the $xy$-plane. Hence, the vector space $A$ corresponds to physical qubits located on $xy$-faces.

The correspondences between the remaining four vector spaces and the physical qubits can be obtained in the same manner. Consequently, the correspondences between the five vector spaces and the five types of physical qubits are summarized in Table \ref{4D Chamon code qubit correspondence}.
\begin{table*}[htbp]
	\centering
	\caption{The correspondence between the types of qubits in 4D Chamon code and vector spaces}
	\label{4D Chamon code qubit correspondence}
	
	\vspace{0.5em}
	
	\makebox[\textwidth][c]{
		\begin{tabular}{c c c c}
			\hline
			Vector space & Tensor-product structure & Coordinate parity & Qubit type \\
			\hline
			$A$ & $C_{-1}\otimes\widetilde{C}_1$ 
			& $(odd,odd,even,even)$ 
			& $xy$-face \\
			\hline
			
			$B$ & $C_{-1}\otimes\widetilde{C}_{-1}$ 
			& $(odd,odd,odd,odd)$ 
			& Hypercube \\
			\hline
			
			$C$ & $C_{0}\otimes\widetilde{C}_{0}$ 
			& two odd numbers and two even numbers 
			& $xz$-, $xw$-, $yz$-, $yw$-face \\
			\hline
			
			$D$ & $C_{1}\otimes\widetilde{C}_{1}$ 
			& $(even,even,even,even)$ 
			& vertex \\
			\hline
			
			$E$ & $C_{1}\otimes\widetilde{C}_{-1}$ 
			& $(even,even,odd,odd)$ 
			& $zw$-face \\
			\hline
		\end{tabular}
	}
\end{table*}

The geometric distribution of the stabilizers can be determined in the same manner. According to the tensor-product structure shown in Fig.~\ref{4DXYZproduct}, the stabilizers fall into four types, corresponding to the vector spaces $S$, $T$, $U$, and $V$, respectively. Based on the tensor-product structure and the parity patterns of their geometric-center coordinates, the correspondence between the four stabilizer types and the associated vector spaces is summarized in Table~\ref{4D Chamon code stabilizers correspondence}.

\begin{table*}[htbp]
	\centering
	\caption{The correspondence between the types of stabilizers in 4D Chamon code and vector spaces}
	\label{4D Chamon code stabilizers correspondence}
	
	\vspace{0.5em}
	
	\makebox[\textwidth][c]{
		\begin{tabular}{c c c c}
			\hline
			Vector space & Tensor-product structure & Coordinate parity & Stabilizer type \\
			\hline
			\multirow{2}{*}{$S$} & \multirow{2}{*}{$C_{-1}\otimes\widetilde{C}_0$} & $(odd,odd,even,odd)$ & \multirow{2}{*}{Cubes perpendicular to $z$- and $w$-axes}\\
			&                                                 & or $(odd,odd,odd,even)$ &\\
			\hline
			\multirow{2}{*}{$T$} & \multirow{2}{*}{$C_{0}\otimes\widetilde{C}_{1}$} & $(odd,even,even,even)$ & \multirow{2}{*}{Edges along $x$- and $y$-axes}\\
			&                                                 &or $(even,odd,even,even)$ &\\
			\hline
			\multirow{2}{*}{$U$} & \multirow{2}{*}{$C_{0}\otimes\widetilde{C}_{-1}$} & $(odd,even,odd,odd)$  & \multirow{2}{*}{Cubes perpendicular to $x$- and $y$-axes}\\
			&                                                 &or $(even,odd,odd,odd)$ &\\
			\hline
			\multirow{2}{*}{$V$} & \multirow{2}{*}{$C_{1}\otimes\widetilde{C}_{0}$} & $(even,even,odd,even)$ & \multirow{2}{*}{Edges along $z$- and $w$-axes}\\
			&                                                 &or $(even,even,even,odd)$ &\\
			\hline
		\end{tabular}
	}
\end{table*}

In summary, the physical qubits of the 4D Chamon code are located at the vertices, the centers of the six types of faces, and the centers of the hypercubes of $\Lambda_4$, whereas the stabilizers are associated with the four types of edges and the four types of cubes. Formally, the 4D Chamon Code are defined as follows:

\begin{definition}[\textbf{The geometric structure and Hamiltonian of the 4D Chamon code}]
	\label{4D Chamon code stabilizer equation}
	The qubits of the 4D Chamon code are located at the vertices, the centers of the hypercubes, and the centers of the six types of faces of the 4D periodic lattice. The X-type stabilizers are associated with the edges along the $x$- and $y$-axes and the cubes perpendicular to the $x$- and $y$-axes, namely,
	\begin{equation}
		S_{e}^{X} = \prod_{v\in V(e)}X_v \prod_{f\in F(e)}X_f,\ 
		S_{c}^{X} = \prod_{hc\in HC(c)}X_{hc} \prod_{f\in F(c)}X_f
	\end{equation}
	Similarly, the $Z$-type stabilizers are associated with the edges along the $z$- and $w$-axes and the cubes perpendicular to the $z$- and $w$-axes, namely,
	\begin{equation}
		S_{e}^{Z} = \prod_{v\in V(e)}Z_v \prod_{f\in F(e)}Z_f,\ 
		S_{c}^{Z} = \prod_{hc\in HC(c)}Z_{hc} \prod_{f\in F(c)}Z_f
	\end{equation}
	Here, $V(e)$ denotes the set of vertices adjacent to edge $e$, $F(e)$ denotes the set of faces adjacent to $e$, $F(c)$ denotes the set of faces contained in the boundary of the cube $c$, and $HC(c)$ denotes the set of hypercubes adjacent to $c$.
	
	The corresponding Hamiltonian is defined as:
	\begin{equation}
		\label{4D Chamon code Hamiltonian}
		H_{Chamon}=-\sum_{e\in E_{xy}}S_e^{X}-\sum_{c\in C_{xy}}S_c^{X}-\sum_{e\in E_{zw}}S_e^{Z}-\sum_{c\in C_{zw}}S_c^{Z}
	\end{equation}
	Here, $E_{xy}$, $C_{\perp xy}$, $E_{zw}$, and $C_{\perp zw}$ denote the set of edges along the $x$- and $y$-axes, the set of cubes perpendicular to the $x$- and $y$-axes, the set of edges along the $z$- and $w$-axes, and the set of cubes perpendicular to the $z$- and $w$-axes, respectively.
\end{definition}

Fig. \ref{4D Chamon code stabilizer} shows the geometric structure  the $X$-type and $Z$-type stabilizers of the 4D Chamon code corresponding to \textbf {Definition} \ref {4D Chamon code stabilizer equation}.

\begin{figure}[htbp]
	\centering
	\includegraphics[width=0.48\textwidth]{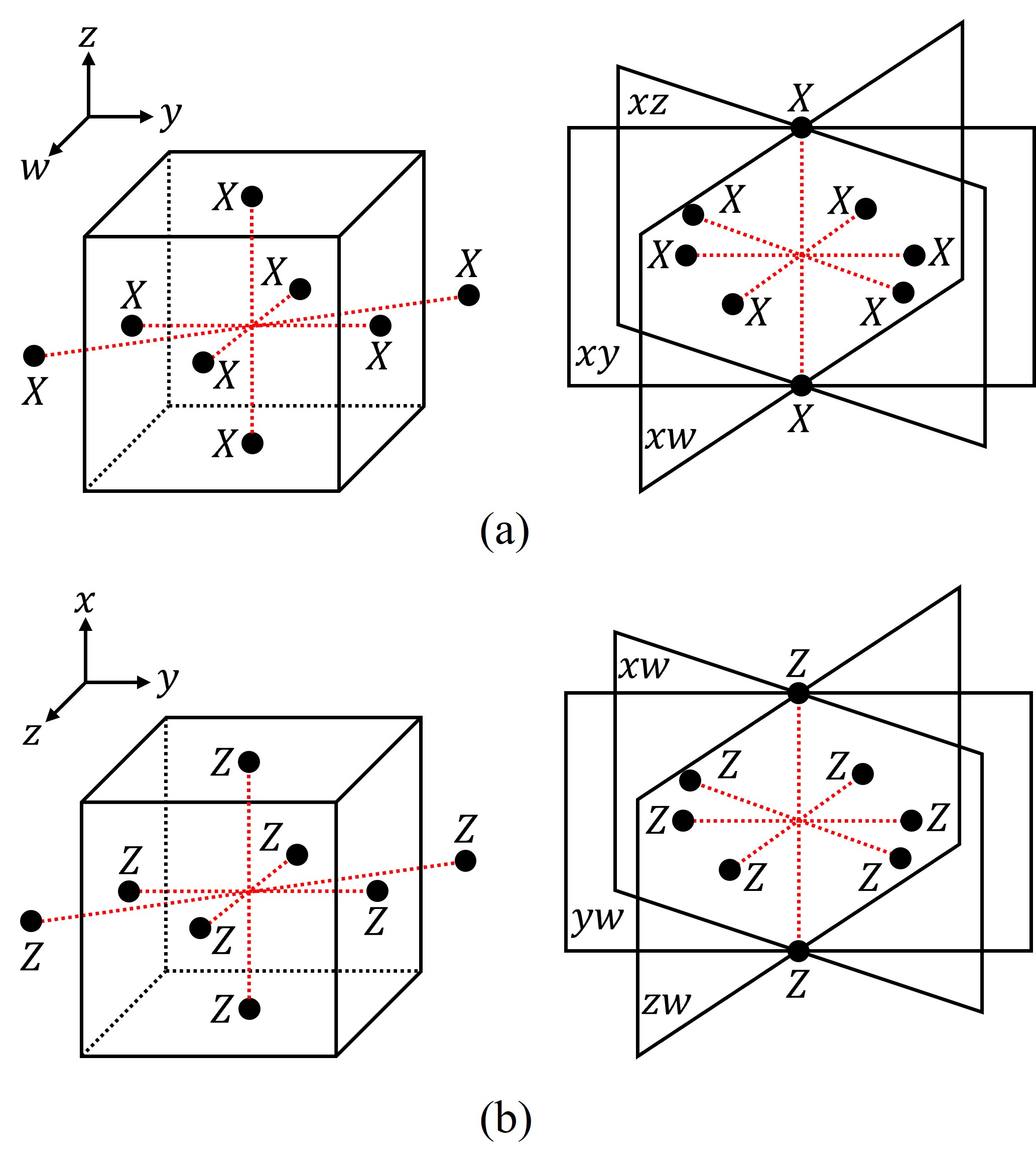}
	\caption{The stabilizers of the 4D Chamon code. (a) The $X$-type stabilizers that corresponds to a cube perpendicular to $x$-axis (left) and an edge along $x$-axis (right). (b) The $Z$-type stabilizers that corresponds to a cube perpendicular to $w$-axis (left) and an edge along $w$-axis (right).}
	\label{4D Chamon code stabilizer}
\end{figure}

In a 4D periodic lattice, each cube is bounded by six faces. Meanwhile, each cube is shared by two adjacent hypercubes, in the same way that each face is shared by two adjacent cubes. Moreover, each edge connects two vertices and is adjacent to six faces. Therefore, each stabilizer shown in Fig. \ref{4D Chamon code stabilizer} has weight $8$.

It can be seen that, regardless of whether the qubit is located at a vertex, a face center, or a hypercube center, a single Pauli $X$ or $Z$ error flips four adjacent stabilizers. There are eight cases in total, as summarized in Table \ref{qubit flip stabilizer types}. Let $E_{\omega}$ denote an edge-type stabilizer along the $\omega$-axis and $C_{\perp \omega}$ denote a cube-type stabilizer perpendicular to the $\omega$-axis, with $\omega\in\{x,y,z,w\}$. From Table \ref{qubit flip stabilizer types}, one can see that Pauli $Z$ errors on vertex qubits and $xy$-face qubits only excite edge stabilizers along the $x$- and $y$-axes; Pauli $Z$ errors on $zw$-face qubits and hypercube-center qubits only excite cube stabilizers perpendicular to the $x$- and $y$-axes; while errors on face qubits lying in the $xz$-, $xw$-, $yz$-, and $yw$- planes excite both edge stabilizers and cube stabilizers simultaneously.

\begin{table*}[htbp]
	\centering
	\caption{The types of stabilizers flipped by qubit errors}
	\label{qubit flip stabilizer types}
	\vspace{0.5em}
	\makebox[\textwidth][c]{
		\begin{tabular}{c c c}
			\hline
			Qubit type & Stabilizers flipped by Pauli $X$ error & Stabilizers flipped by Pauli $Z$ error\\
			\hline
			Vertex & $2E_z+2E_w$ & $2E_x+2E_y$\\
			\hline
			Hypercube & $2C_{\perp z}+2C_{\perp w}$ & $2C_{\perp x}+2C_{\perp y}$\\
			\hline
			$xy$-face & $2C_{\perp z}+2C_{\perp w}$ & $2E_x+2E_y$\\
			\hline
			$zw$-face & $2E_z+2E_w$ & $2C_{\perp x}+2C_{\perp y}$\\
			\hline
			$xz$-face & $2E_z+2C_{\perp w}$ & $2E_x+2C_{\perp y}$ \\
			\hline
			$xw$-face & $2C_{\perp z}+2E_w$ & $2E_x+2C_{\perp y}$ \\
			\hline
			$yz$-face & $2E_z+2C_{\perp w}$ & $2C_{\perp x}+2E_y$ \\
			\hline
			$yw$-face & $2C_{\perp z}+2E_w$ & $2C_{\perp x}+2E_y$ \\
			\hline
		\end{tabular}
	}
\end{table*}

\subsection{Excitations}
In our previous work \cite{liang2025high}, we proved that the 4D Chamon code encodes $k=8\gcd(L_x,L_y)\gcd(L_z,L_w)$ logical qubits. Consequently, the ground-state degeneracy (GSD) of the Hamiltonian in Eq.~(\ref{4D Chamon code Hamiltonian}) is $2^k$ and therefore grows exponentially with the lattice size. Such system-size-dependent GSD is an important characteristic of fracton models. However, this characteristic alone is not sufficient to establish the 4D Chamon code as a new 4D fracton model.

In this section, we show that the 4D Chamon code possesses three types of excitations with restricted mobility, namely, monopoles, dipoles, and quadrupoles. The monopole and dipole excitations have the same structures as their counterparts in the 3D Chamon code, while the quadrupole excitations in the two models exhibit similar structural features. These results further support the interpretation of the 4D Chamon code as a 4D generalization of the 3D Chamon code and as a new 4D fracton model.

Since the $X$-type and $Z$-type stabilizers of the 4D Chamon code have analogous structures, in this paper we focus only on the excitations created by Pauli $Z$ errors that flip the measurement outcomes of $X$-type stabilizers. The corresponding conclusions for the excitations generated by Pauli $X$ errors that flip the measurement outcomes of $Z$-type stabilizers can be drawn analogously.

\subsubsection{Monopole and dipole}
\label{Monopole and dipole}
In this section, we describe two types of operators that create monopole and dipole excitations, respectively.

Note that, in a 4D lattice, if the geometric centers of several objects share the same $z$- and $w$-coordinates, then they lie in the same $xy$-plane. More formally, for any fixed values $z_0$ and $w_0$, the set
\begin{equation}
	\Pi_{xy}(z_0,w_0)
	=
	\{(x,y,z,w)\mid z=z_0,\ w=w_0\}
\end{equation}
defines a 2D $xy$-plane. Hence, the geometric centers of all objects with the same $z$- and $w$-coordinates lie in plane $\Pi_{xy}(z_0,w_0)$.

Similarly, for fixed values $x_0$ and $y_0$, the set
\begin{equation}
	\Pi_{zw}(x_0,y_0)
	=
	\{(x,y,z,w)\mid x=x_0,\ y=y_0\}
\end{equation}
defines a 2D $zw$-plane containing the geometric centers of all objects with the same $x$- and $y$-coordinates.

In the 4D Chamon code, isolated monopole excitations can be created by rectangular membrane operators supported on 2D $xy$- or $zw$-planes. For example, within an $xy$-plane $\Pi_{xy}(z_0,w_0)$, consider a rectangular region $M=M_{\mathrm{odd}}\cup M_{\mathrm{even}}$, highlighted in orange in Fig.~\ref{Monopoles and Dipoles}(a), where $M_{\mathrm{odd}}$ and $M_{\mathrm{even}}$ denote the sets of geometric centers of objects for which $x+y$ is odd and even, respectively. The corresponding rectangular membrane operator is defined as
\begin{equation}
	W(M)=\prod_{m\in M_{\mathrm{even}}} Z_m .
\end{equation}
It can be verified that $W(M)$ anticommutes only with the $X$-type stabilizers located near the four corners of the rectangular region. Consequently, it flips the measurement outcomes of these four stabilizers and creates four isolated monopole excitations, as indicated by the purple double circles in Fig.~\ref{Monopoles and Dipoles}(a).

In the 4D Chamon code, dipole excitations can be created by rigid string operators supported on 2D $xy$- or $zw$-planes. For example, within an $xy$-plane $\Pi_{xy}(z_0,w_0)$, consider a string-like region aligned with the face diagonal, denoted by $\gamma=\gamma_{\mathrm{odd}}\cup\gamma_{\mathrm{even}}$, as highlighted in orange in Fig.~\ref{Monopoles and Dipoles}(b), where $\gamma_{\mathrm{odd}}$ and $\gamma_{\mathrm{even}}$ denote the sets of geometric centers of objects for which $x+y$ is odd and even, respectively. The corresponding rigid string operator is defined as
\begin{equation}
	W(\gamma)=\prod_{m\in\gamma_{\mathrm{even}}} Z_m .
\end{equation}
It can be verified that $W(\gamma)$ anticommutes with a pair of $X$-type stabilizers located near each endpoint of the string. Consequently, it flips the measurement outcomes of four stabilizers in total and creates two pairs of dipoles, as indicated by the purple double circles in Fig.~\ref{Monopoles and Dipoles}(b).

\begin{figure}[htbp]
	\centering
	\includegraphics[width=0.48\textwidth]{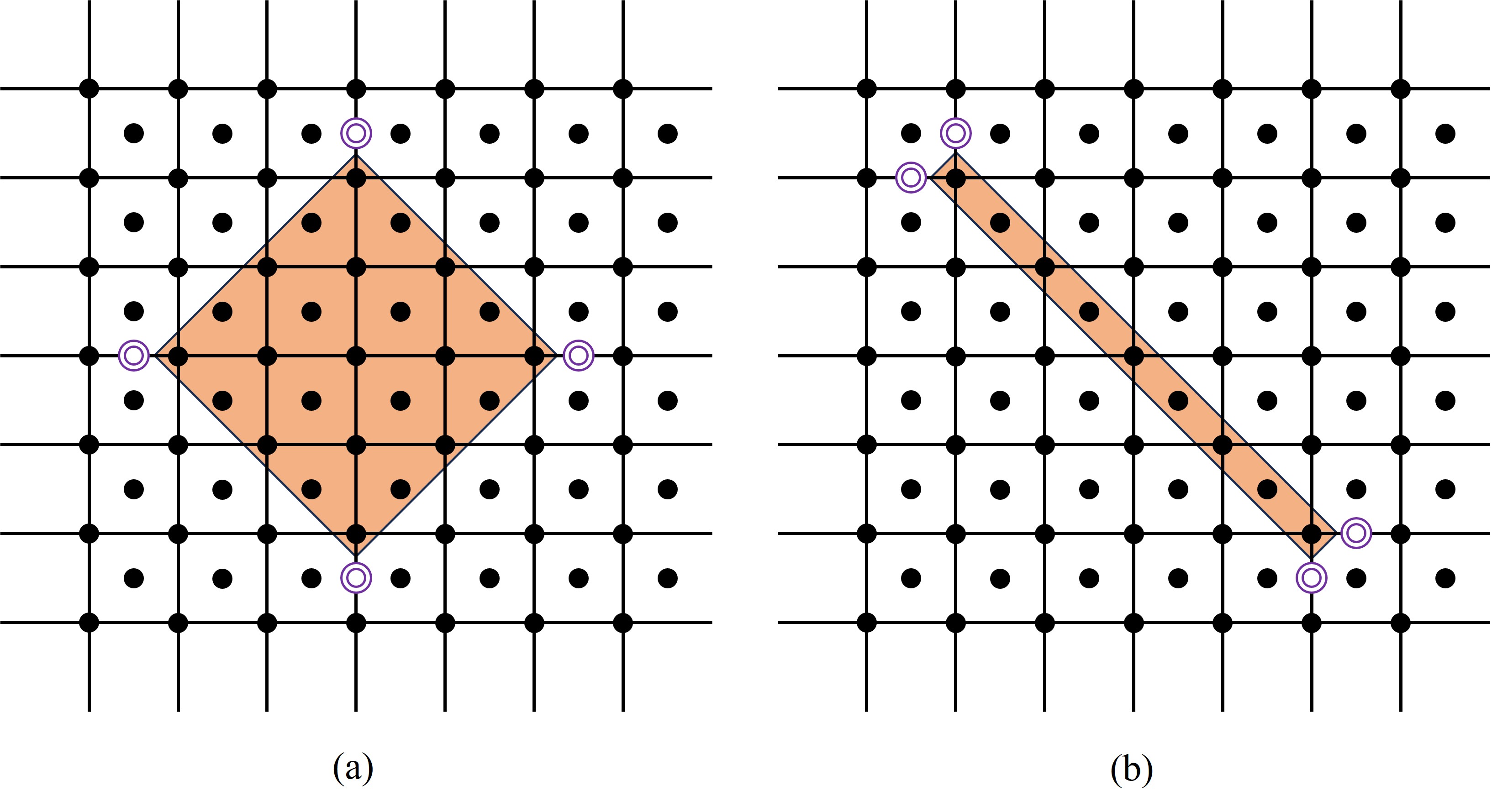}
	\caption{The (a) monopoles and (b) dipoles in the 4D Chamon code, which are created by rectangular-shaped membrane operator and rigid string operator, respectively.}
	\label{Monopoles and Dipoles}
\end{figure}

Unlike the freely mobile excitations in conventional topological order models, the excitations in the 4D Chamon code have restricted mobility. A single monopole cannot be moved individually by any finitely supported Pauli operator; any local operation attempting to move a single monopole will create additional excitations. For a dipole, although it is not completely immobile like a monopole, it can only be moved by a Pauli operator with finite support along the direction of the rigid string, namely, it can move only along the direction of the string.

\subsubsection{Quadrupole}
\label{Quadrupole}
Unlike the dipole excitations created by rigid string operators, Ref.~\cite{bravyi2011topological} demonstrated that quadrupole excitations in the 3D Chamon code can be created by flexible string operators. Such flexible strings are supported on two adjacent planes perpendicular to a body-diagonal. They can be deformed arbitrarily within these two planes but cannot extend beyond them. See Ref.~\cite{bravyi2011topological} for further details.

In the following, we demonstrate that the 4D Chamon code admits an analogous structure: its quadrupole excitations can be generated by flexible ribbon operators supported on two adjacent hyperplanes. To provide a clearer geometric representation of these flexible ribbon operators, we first show that the stabilizers of the 4D Chamon code located on hyperplanes perpendicular to the vector $(1,1,1,1)$ can be projected onto a 3D face-centered cubic (FCC) lattice.

First, we introduce the following equation for a family of hyperplanes:
\begin{equation}
	(r_x x+r_y y+r_z z+r_w w)\bmod g = E,
\end{equation}
where $\mathbf{r}=(r_x,r_y,r_z,r_w)$ denotes the normal vector of the hyperplanes, and
\begin{equation}
	g=\gcd(2r_xL_x,2r_yL_y,2r_zL_z,2r_wL_w).
\end{equation}
According to the geometric representation introduced in Sec.~\ref{Geometric notations}, all physical qubits lie on hyperplanes with normal vector $(1,1,1,1)$ and even values of $E$, whereas all stabilizers lie on hyperplanes with odd values of $E$.

Consider the hyperplane $\Sigma_e$ with normal vector $\mathbf{r}=(1,1,1,1)$ and $E=e$. Any coordinate $\mathbf{q}=(a,b,c,d)$ lying on $\Sigma_e$ satisfies
\begin{equation}
	(a+b+c+d)\bmod d=e.
\end{equation}
This constraint leaves three independent degrees of freedom, and hence $\Sigma_e$ is three-dimensional. This is analogous to a plane in a 3D cubic lattice that is perpendicular to the body-diagonal $(1,1,1)$ and therefore has two independent degrees of freedom. Consequently, the stabilizers lying on $\Sigma_e$, together with the qubits lying on the two neighboring hyperplanes $\Sigma_{e+1}$ and $\Sigma_{e-1}$, can be represented in a 3D coordinate system through an isometric projection.

To construct a 3D coordinate system associated with the hyperplane $\Sigma_e$, we choose three unit vectors orthogonal to the normal vector $(1,1,1,1)$:
\begin{equation}
	\label{basis}
	\begin{aligned}
		&n_1 = \frac{1}{2}(1,1,-1,-1),\\
		&n_2 = \frac{1}{2}(1,-1,1,-1),\\
		&n_3 = \frac{1}{2}(1,-1,-1,1).
	\end{aligned}
\end{equation}
It is straightforward to verify that
$n_i\cdot n_j=\delta_{ij}$ and
$n_i\cdot(1,1,1,1)=0$.
Therefore, $\{n_1,n_2,n_3\}$ forms an orthonormal basis for the subspace orthogonal to the normal vector $(1,1,1,1)$.

Taking the reference point $r_0=(e,0,0,0)$ on $\Sigma_e$ as the origin of the 3D coordinate system, the 3D coordinate of an arbitrary point
$r=(a,b,c,d)\in\Sigma_e$ is defined by the mapping
\begin{equation}
	\label{projection1}
	\Phi_e(r)=
	\begin{bmatrix}
		n_1\cdot(r-r_0)\\
		n_2\cdot(r-r_0)\\
		n_3\cdot(r-r_0)
	\end{bmatrix}.
\end{equation}
Explicitly, we obtain
\begin{equation}
	\label{projection2}
	\Phi_e(r)=
	\begin{bmatrix}
		\xi\\
		\eta\\
		\zeta
	\end{bmatrix}
	=
	\begin{bmatrix}
		\frac{a+b-c-d-e}{2}\\
		\frac{a-b+c-d-e}{2}\\
		\frac{a-b-c+d-e}{2}
	\end{bmatrix}.
\end{equation}

Since the three basis vectors in Eq.~(\ref{basis}) are mutually orthogonal and have unit norm, the coordinate transformation defined in Eq.~(\ref{projection1}) preserves Euclidean distances within the hyperplane $\Sigma_e$. Hence, $\Phi_e$ defines an isometry between $\Sigma_e$ and its 3D coordinate representation, rather than a geometrically distorted projection.

For a stabilizer located on the hyperplane $\Sigma_e$ with coordinate
$r_s=(a_s,b_s,c_s,d_s)$, Eq.~\eqref{projection2} gives
\begin{equation}
	\label{projection3}
	\begin{bmatrix}
		\xi\\
		\eta\\
		\zeta
	\end{bmatrix}
	=
	\begin{bmatrix}
		a_s+b_s-e\\
		a_s+c_s-e\\
		-b_s-c_s
	\end{bmatrix}.
\end{equation}
It follows that
\begin{equation}
	\xi+\eta+\zeta=2(a_s-e),
\end{equation}
and hence
\begin{equation}
	(\xi+\eta+\zeta)\bmod 2=0.
\end{equation}
Therefore, the projected coordinates of the stabilizers on $\Sigma_e$
form the 3D lattice
\begin{equation}
	\Lambda_3
	=
	\left\{
	(\xi,\eta,\zeta)\in\mathbb{Z}_{2L}^3
	\;\middle|\;
	(\xi+\eta+\zeta)\bmod 2=0
	\right\}.
\end{equation}
This lattice is equivalent to a 3D FCC lattice, in which the projected stabilizers are located at the eight vertices and six face centers of each conventional cubic unit cell. We next characterize the distribution patterns of the $X$-type and $Z$-type stabilizers on $\Sigma_e$ after projection onto the FCC lattice.

\begin{corollary}[\textbf{Distribution of projected stabilizers}]
	\label{X and Z stabilizer projection}
	After the $X$-type and $Z$-type stabilizers on the hyperplane $\Sigma_e$ are projected onto the 3D FCC lattice $\Lambda_3$, the $X$-type stabilizers are located at the vertices and the centers of the $\eta\zeta$-faces, whereas the $Z$-type stabilizers are located at the centers of the $\xi\zeta$- and $\xi\eta$-faces. Equivalently,
	\begin{equation}
		\begin{aligned}
			\Lambda_X
			&=
			\{(\xi_X,\eta_X,\zeta_X)\in\Lambda_3
			\mid \xi_X \bmod 2 =0 \},\\
			\Lambda_Z
			&=
			\{(\xi_Z,\eta_Z,\zeta_Z)\in\Lambda_3
			\mid \xi_Z \bmod 2 =1 \}.
		\end{aligned}
	\end{equation}
\end{corollary}

\begin{proof}
	Combining the correspondence between the stabilizer types of the 4D Chamon code and the parity patterns of their 4D geometric-center coordinates, as summarized in Table~\ref{4D Chamon code stabilizers correspondence}, with Eq.~\eqref{projection3}, we obtain the correspondence between the stabilizer types and the parity patterns of their projected coordinates in the 3D FCC lattice, as summarized in Table~\ref{4D Chamon code stabilizer FCC correspondence}.
	
	\begin{table}[htbp]
		\caption{Correspondence between the stabilizer types of the 4D Chamon code and the projected coordinates in the 3D FCC lattice}
		\label{4D Chamon code stabilizer FCC correspondence}
		\begin{tabular}{c c c}
			\hline
			Stabilizer type & $(\xi,\eta,\zeta)\bmod 2$ & FCC sublattice\\
			\hline
			$E_x$ & \multirow{2}{*}{$(0,0,0)$} & \multirow{2}{*}{$\Lambda_x$}\\
			$C_{\perp x}$ & & \\
			\hline
			$E_y$ & \multirow{2}{*}{$(0,1,1)$} & \multirow{2}{*}{$\Lambda_y$}\\
			$C_{\perp y}$ & & \\
			\hline
			$E_z$ & \multirow{2}{*}{$(1,0,1)$} & \multirow{2}{*}{$\Lambda_z$}\\
			$C_{\perp z}$ & & \\
			\hline
			$E_w$ & \multirow{2}{*}{$(1,1,0)$} & \multirow{2}{*}{$\Lambda_w$}\\
			$C_{\perp w}$ & & \\
			\hline
		\end{tabular}
	\end{table}
	
	Therefore, the FCC lattice $\Lambda_3$ decomposes into four mutually disjoint sublattices,
	\begin{equation}
		\Lambda_3
		=
		\Lambda_x
		\mathbin{\dot\cup}
		\Lambda_y
		\mathbin{\dot\cup}
		\Lambda_z
		\mathbin{\dot\cup}
		\Lambda_w .
	\end{equation}
	Moreover,
	\begin{equation}
		\Lambda_X=\Lambda_x\cup\Lambda_y,\ \Lambda_Z=\Lambda_z\cup\Lambda_w.
	\end{equation}
	It therefore follows that
	\begin{equation}
		\begin{aligned}
			\Lambda_X
			&=
			\{(\xi_X,\eta_X,\zeta_X)\in\Lambda_3
			\mid \xi_X \bmod 2 =0 \},\\
			\Lambda_Z
			&=
			\{(\xi_Z,\eta_Z,\zeta_Z)\in\Lambda_3
			\mid \xi_Z \bmod 2 =1 \},
		\end{aligned}
	\end{equation}
	which completes the proof.
\end{proof}

According to Corollary \ref{X and Z stabilizer projection}, the $X$-type and $Z$-type stabilizers exhibit a layered alternating distribution along the $\xi$-axis of the FCC lattice.

For a qubit located on either $\Sigma_{e+1}$ or $\Sigma_{e-1}$ with coordinate
$r_q=(a_q,b_q,c_q,d_q)$, substituting $r_q$ into Eq.~\eqref{projection1} gives
\begin{equation}
	\label{projection4}
	\Phi_e(r_q)=
	\begin{bmatrix}
		\xi_q\\
		\eta_q\\
		\zeta_q
	\end{bmatrix}
	=
	\begin{bmatrix}
		a_q+b_q-\left(e+\frac{\sigma}{2}\right)\\
		a_q+c_q-\left(e+\frac{\sigma}{2}\right)\\
		-b_q-c_q+\frac{\sigma}{2}
	\end{bmatrix},
\end{equation}
where $\sigma=+1$ for $r_q\in\Sigma_{e+1}$ and $\sigma=-1$ for
$r_q\in\Sigma_{e-1}$.

Eq.~\eqref{projection4} shows that all three components of the projected
coordinate $(\xi_q,\eta_q,\zeta_q)$ are half-integers. Geometrically, these
coordinates correspond to the centers of the eight unit cubes contained
within a cubic unit cell of side length 2. Therefore, the projected coordinate
can be written as
\begin{equation}
	(\xi_q,\eta_q,\zeta_q)
	=
	\left(
	i+\frac{1}{2},
	j+\frac{1}{2},
	k+\frac{1}{2}
	\right),
	\ i,j,k\in\mathbb{Z}.
\end{equation}
We denote the qubit located at this projected coordinate by $q_{i,j,k}$.

Next, we show that the parity of $i+j+k$ determines from which neighboring
hyperplane the qubit originates: if $i+j+k$ is even, then $q_{i,j,k}$
originates from $\Sigma_{e+1}$; if $i+j+k$ is odd, then it originates from
$\Sigma_{e-1}$.

\begin{corollary}[\textbf{Distribution of projected qubits}]
	\label{qubit projection}
	For a projected qubit with coordinate
	\[
	(\xi_q,\eta_q,\zeta_q)
	=
	\left(
	i+\frac{1}{2},
	j+\frac{1}{2},
	k+\frac{1}{2}
	\right),
	\]
	the qubit originates from the hyperplane $\Sigma_{e+1}$ if
	$(i+j+k)\bmod 2=0$, whereas it originates from the hyperplane
	$\Sigma_{e-1}$ if $(i+j+k)\bmod 2=1$.
\end{corollary}

\begin{proof}
	From Eq.~\eqref{projection4}, we obtain
	\begin{equation}
		\begin{aligned}
			\xi_q+\eta_q+\zeta_q
			&=
			\left(a_q+b_q-e-\frac{\sigma}{2}\right)
			+
			\left(a_q+c_q-e-\frac{\sigma}{2}\right)\\
			&\quad+
			\left(-b_q-c_q+\frac{\sigma}{2}\right)\\
			&=
			2a_q-2e-\frac{\sigma}{2}.
		\end{aligned}
	\end{equation}
	On the other hand, since
	$\xi_q=i+\frac{1}{2}$,
	$\eta_q=j+\frac{1}{2}$, and
	$\zeta_q=k+\frac{1}{2}$, we have
	\begin{equation}
		\xi_q+\eta_q+\zeta_q
		=
		i+j+k+\frac{3}{2}.
	\end{equation}
	Equating the above two expressions gives
	\begin{equation}
		i+j+k
		=
		2(a_q-e)-\frac{\sigma+3}{2}.
	\end{equation}
	If the qubit lies on $\Sigma_{e+1}$, then $\sigma=1$, and hence
	\begin{equation}
		i+j+k=2(a_q-e)-2,
	\end{equation}
	which is even. If the qubit lies on $\Sigma_{e-1}$, then $\sigma=-1$, and hence
	\begin{equation}
		i+j+k=2(a_q-e)-1,
	\end{equation}
	which is odd. This completes the proof.
\end{proof}

According to Corollary~\ref{qubit projection}, any two qubits with adjacent projected coordinates must originate from different hyperplanes. Fig.~\ref{projection_fig} shows the 3D lattice obtained by projecting the stabilizers on $\Sigma_e$ together with the qubits on the neighboring hyperplanes $\Sigma_{e+1}$ and $\Sigma_{e-1}$. In the figure, black circles represent stabilizers, red circles represent qubits originating from $\Sigma_{e+1}$, and blue circles represent qubits originating from $\Sigma_{e-1}$. The projected qubits originating from the two neighboring hyperplanes alternate throughout the lattice. For clarity, the stabilizers located at the face centers, as well as those located at the two vertices with coordinates $(0,0,0)$ and $(2,2,2)$, are omitted from Fig.~\ref{projection_fig}.

\begin{figure}[htbp]
	\centering
	\includegraphics[width=0.3\textwidth]{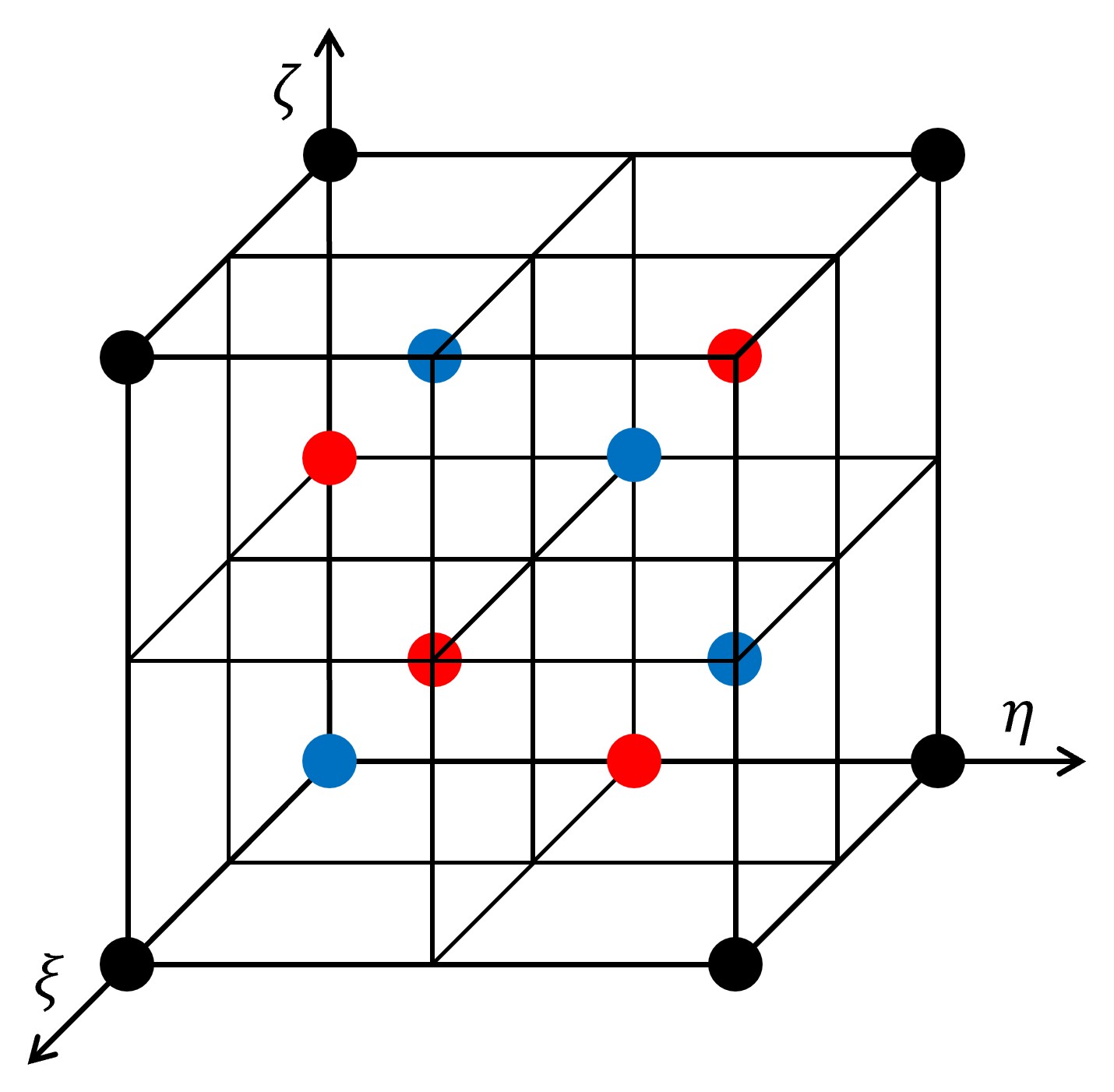}
	\caption{The projected 3D lattice, with black dots representing stabilizers and red dots representing qubits.}
	\label{projection_fig}
\end{figure}

To show that the quadrupole excitations in the 4D Chamon code can be generated by flexible ribbon operators, we first define the following elementary weight-2 rectangular operator:
\begin{equation}
	\label{square operator}
	B_{i,j,k}^Z
	=
	Z_{q_{i,j,k}}Z_{q_{i+1,j,k}},
\end{equation}
where $i$ is required to be even. As discussed above, the qubits
$q_{i,j,k}$ and $q_{i+1,j,k}$ must originate from two different hyperplanes.
Therefore, $B_{i,j,k}^Z$ is a staggered bilayer weight-2 operator.
According to the parities of $j$ and $k$, the operator defined in
Eq.~\eqref{square operator} can be classified into four types, all of which
flip eight $X$-type stabilizers. Table~\ref{weight 2 operators} summarizes
the four types of operators and the corresponding excitation distributions, which are also illustrated in Fig.~\ref{Four rectangular operators}.

\begin{table}[htbp]
	\caption{Four types of weight-2 rectangular operators and the corresponding excitation distributions}
	\label{weight 2 operators}
	\begin{tabular}{c c c}
		\hline
		Rectangular operator & $(j,k)\bmod 2$ & Excitation distribution\\
		\hline
		$B_{00}$ & $(0,0)$ & lower left--upper right\\
		\hline
		$B_{10}$ & $(1,0)$ & upper left--lower right\\
		\hline
		$B_{01}$ & $(0,1)$ & upper left--lower right\\
		\hline
		$B_{11}$ & $(1,1)$ & lower left--upper right\\
		\hline
	\end{tabular}
\end{table}

The purple double circles in Fig.~\ref{Four rectangular operators} indicate
the excitations located at the vertices and at the centers of the
$\eta\zeta$-faces. Note that both the vertex excitations and the
$\eta\zeta$-face excitations lie on the hyperplane $\Sigma_e$.
Each purple triple circle represents two excitations that overlap after
projection onto the 3D lattice, with one originating from $\Sigma_{e+1}$
and the other from $\Sigma_{e-1}$.

When viewed along the $\xi$-axis, the excitation pattern depends on the
parity of $j+k$. If $(j+k)\bmod 2=0$, the excitations are located at the
lower-left and upper-right corners; if $(j+k)\bmod 2=1$, they are located
at the lower-right and upper-left corners. In addition, as shown in
Fig.~\ref{Four rectangular operators}, two horizontally adjacent rectangular
operators of different parity types can be concatenated along the
$\eta$-direction as $B_{00}B_{10}$ or $B_{01}B_{11}$, thereby annihilating
the four internal excitations. Similarly, two vertically adjacent rectangular
operators of different parity types can be concatenated along the
$\zeta$-direction as $B_{00}B_{01}$ or $B_{10}B_{11}$, again annihilating
the four internal excitations.
\begin{figure}[htbp]
	\centering
	\includegraphics[width=0.48\textwidth]{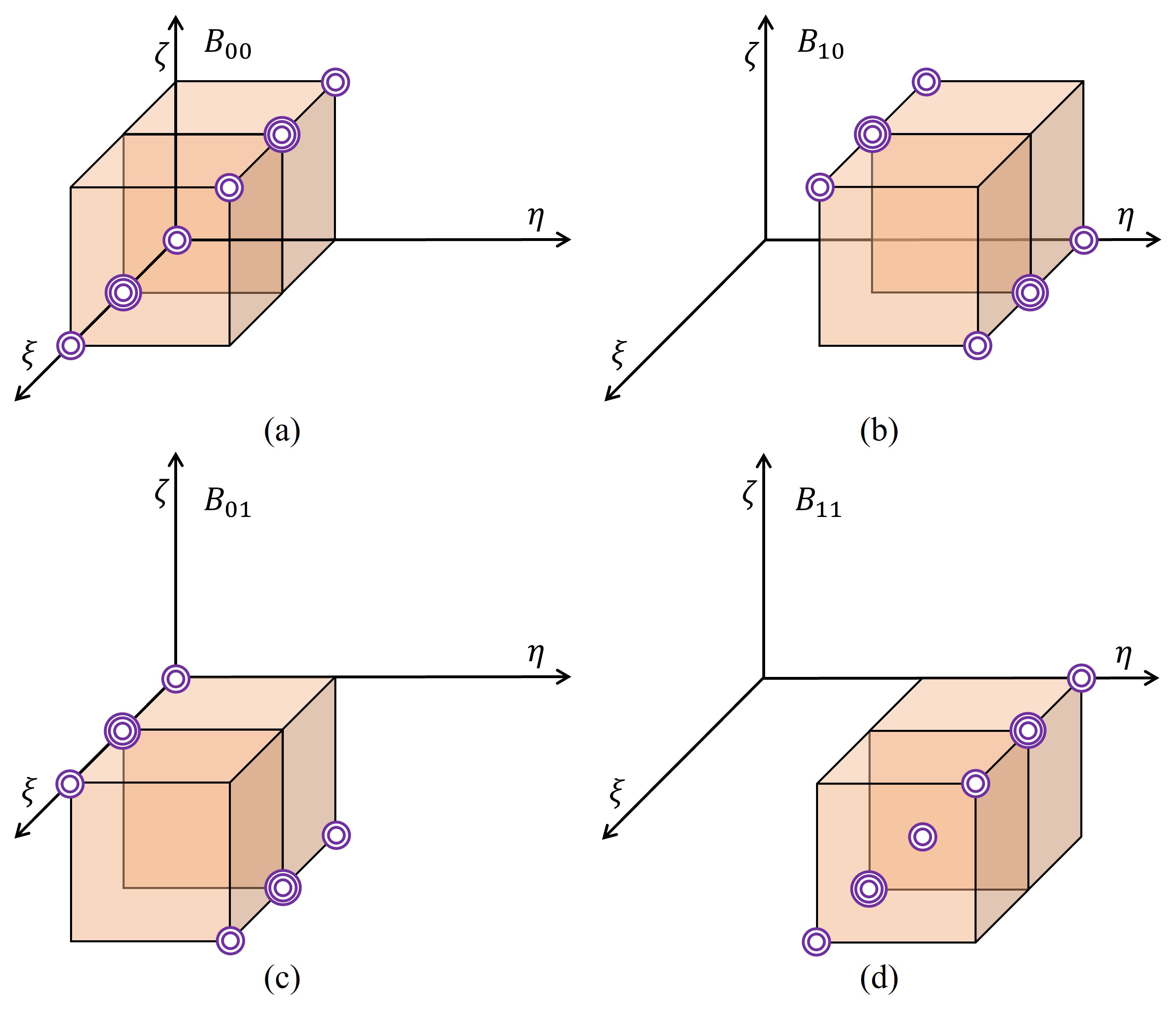}
	\caption{Four types of rectangular operators with weight of 2 and the corresponding excitations.}
	\label{Four rectangular operators}
\end{figure}

Therefore, by concatenating these elementary rectangular operators in this way, one can form a bendable flexible bilayer ribbon operator $W_\gamma$ inside the 3D lattice, as shown in Fig. \ref{Quadrupoles}. The excitations generated by the elementary rectangular operators in the interior of $W_\gamma$ cancel each other, while four excitations remain at each of its two endpoints. The endpoint excitation form a \textbf{quadrupole}. More formally, the flexible bilayer ribbon operator is defined as follows.

\begin{definition}[\textbf{Flexible bilayer ribbon operator}]
	Let $\gamma=(c_0,c_1,\cdots,c_m)$ be a sequence of cell coordinates on a 2D integer lattice, where $c_t=(j_t,k_t)$. If for each $t=0,1,\cdots,m-1$, there exist a direction $\alpha_t\in{\eta,\zeta}$ and a sign $\varepsilon_t\in{+1,-1}$ such that
	$c_{t+1}=c_t+\varepsilon_t\mathbf e_{\alpha_t}$, where $\mathbf e_{\eta}=(1,0)$ and $\mathbf e_{\zeta}=(0,1)$, then $\gamma$ is called a flexible bilayer ribbon path. The corresponding operator
	\begin{equation}
		W_r^Z(\gamma)=\prod\limits_{t=0}^{m}B^Z_{i,j_t,k_t}
	\end{equation}
	is called a flexible bilayer ribbon operator, where $i$ is required to be even.
\end{definition}

\begin{figure}[htbp]
	\centering
	\includegraphics[width=0.4\textwidth]{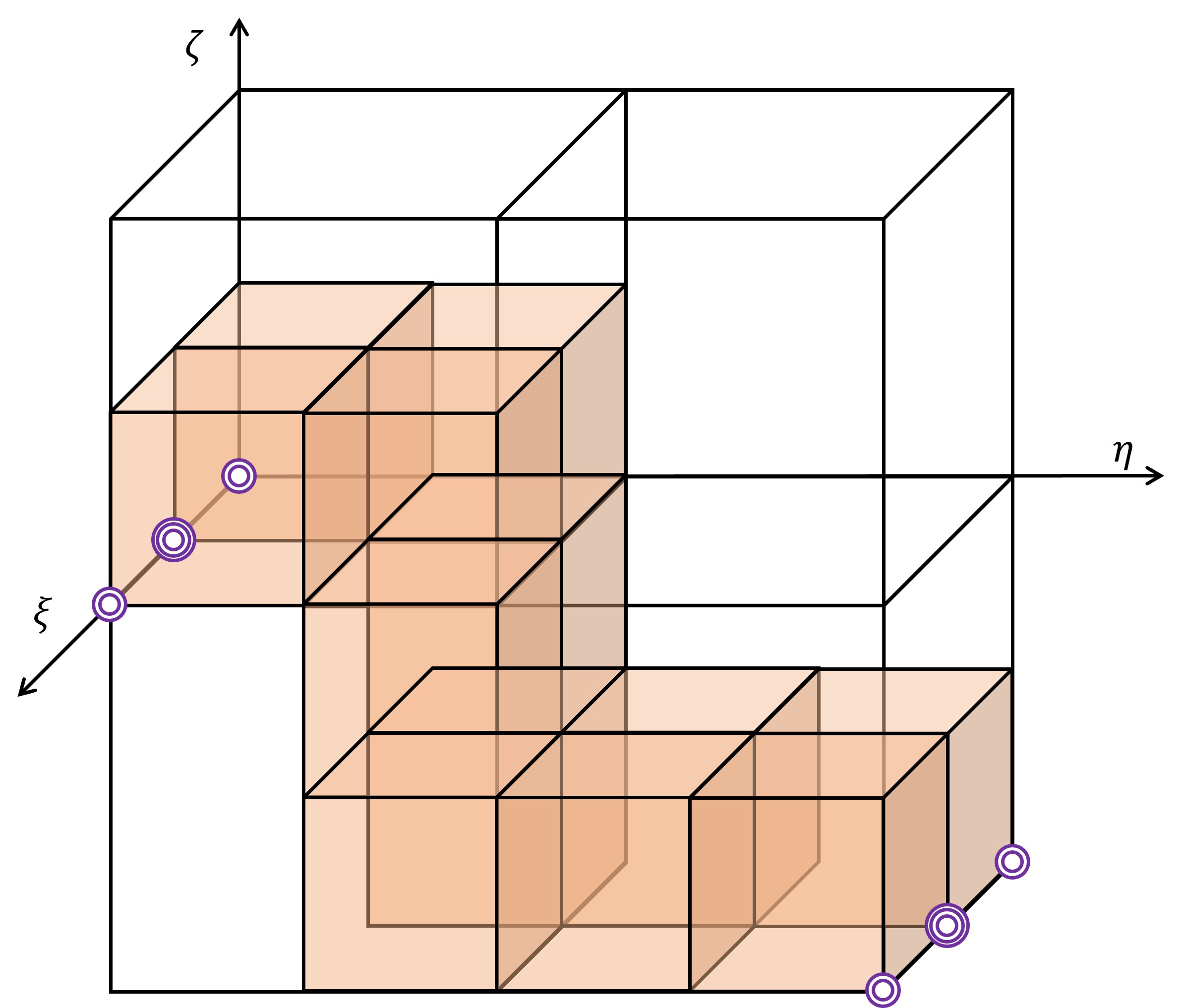}
	\caption{Flexible bilayer ribbon operator and quadrupoles.}
	\label{Quadrupoles}
\end{figure}

\section{The two layers decoding strategy for the 4D Chamon code}
\label{two-layer decoding algorithm}
For a quantum stabilizer code, the most-likely-error decoding problem can be formulated as finding an error estimate $\hat{\mathbf{e}}$ satisfying
\begin{equation}
	\label{IS}
	H\hat{\mathbf{e}}=\mathbf{s}\pmod 2,
\end{equation}
while maximizing the conditional probability
$P(\hat{\mathbf{e}}\mid\mathbf{s})$.
Here, $H$ denotes the parity-check matrix of the stabilizer code,
$\mathbf{s}$ is the syndrome vector, and $\hat{\mathbf{e}}$ is an estimate
of the actual error $\mathbf{e}$. In general, this decoding problem is
NP-hard.

When each column of $H$ has weight at most two, each elementary
error is associated with at most two defects. In this case, the
decoding problem can be mapped to a minimum-weight matching problem on an
ordinary graph and can be solved exactly in polynomial time. Typical
examples include the decoding of surface codes and toric codes.

In contrast, when some columns of $H$ have weight greater than two, a
single elementary error may create more than two defects. A direct generalization of the matching formulation then leads to a minimum-weight matching problem on a hypergraph rather than on an ordinary graph. Such hypergraph matching problems are NP-hard in general, and no general polynomial-time exact algorithm is known.

For the 4D Chamon code, a Pauli $Z$ or $X$ error acting on a physical qubit creates four defects. Therefore, the 4D Chamon code cannot be decoded directly using the conventional MWPM algorithm defined on ordinary graphs. Although BP-OSD is a general-purpose decoding algorithm applicable to a broad class of quantum stabilizer codes, its decoding accuracy for the 4D Chamon code degrades as the code size increases. The goal of this section is therefore to develop a decoding algorithm tailored to the 4D Chamon code that achieves higher decoding accuracy than BP-OSD.

The central idea of the two-layer decoding strategy proposed in this paper is to avoid directly solving the global linear system in Eq.~\eqref{IS}. By jointly exploiting the geometric structure of the 4D Chamon code and the algebraic structure of its parity-check matrix, we show that the original decoding problem can be decomposed into multiple smaller and tractable subproblems. More importantly, these subproblems can be solved independently and therefore processed in parallel, substantially reducing the computational cost of decoding. The solutions of the individual subproblems are then combined to construct a global solution satisfying the original linear system in Eq.~\eqref{IS}. This decomposition relies on two important properties of the 4D Chamon code:

1. \textbf{Symmetry of the distribution of stabilizers}: A Pauli $Z$ or $X$ error acting on a physical qubit flips the measurement outcomes of two stabilizers on each of the two hyperplanes adjacent to the hyperplane containing that qubit. Consequently, the original MWPM problem on hypergraph can be decomposed into a collection of MWPM problems on ordinary graphs, which can be solved independently and in parallel.

2. \textbf{Projection-induced 2D toric-code structure}:  The original syndrome $\mathbf{s}$ can be projected either onto $xy$-planes with fixed $(z,w)$ coordinates or onto $zw$-planes with fixed $(x,y)$ coordinates. Each projected subproblem is equivalent to the problem of decoding Pauli $Y$ errors in a 2D toric code. Therefore, through these projections, the original decoding problem can be decomposed into multiple independent subproblems, each of which can be solved in parallel using a decoding procedure equivalent to that for Pauli $Y$ errors in a 2D toric code.

In the following, these two properties together with the corresponding decoding procedures that exploit them are introduced in detail, thereby leading to the double-layer decoding strategy proposed in this paper. The overall architecture of the strategy is illustrated in Fig. \ref{DD_decoder_flowchart}.
\begin{figure}[htbp]
	\centering
	\includegraphics[width=0.48\textwidth]{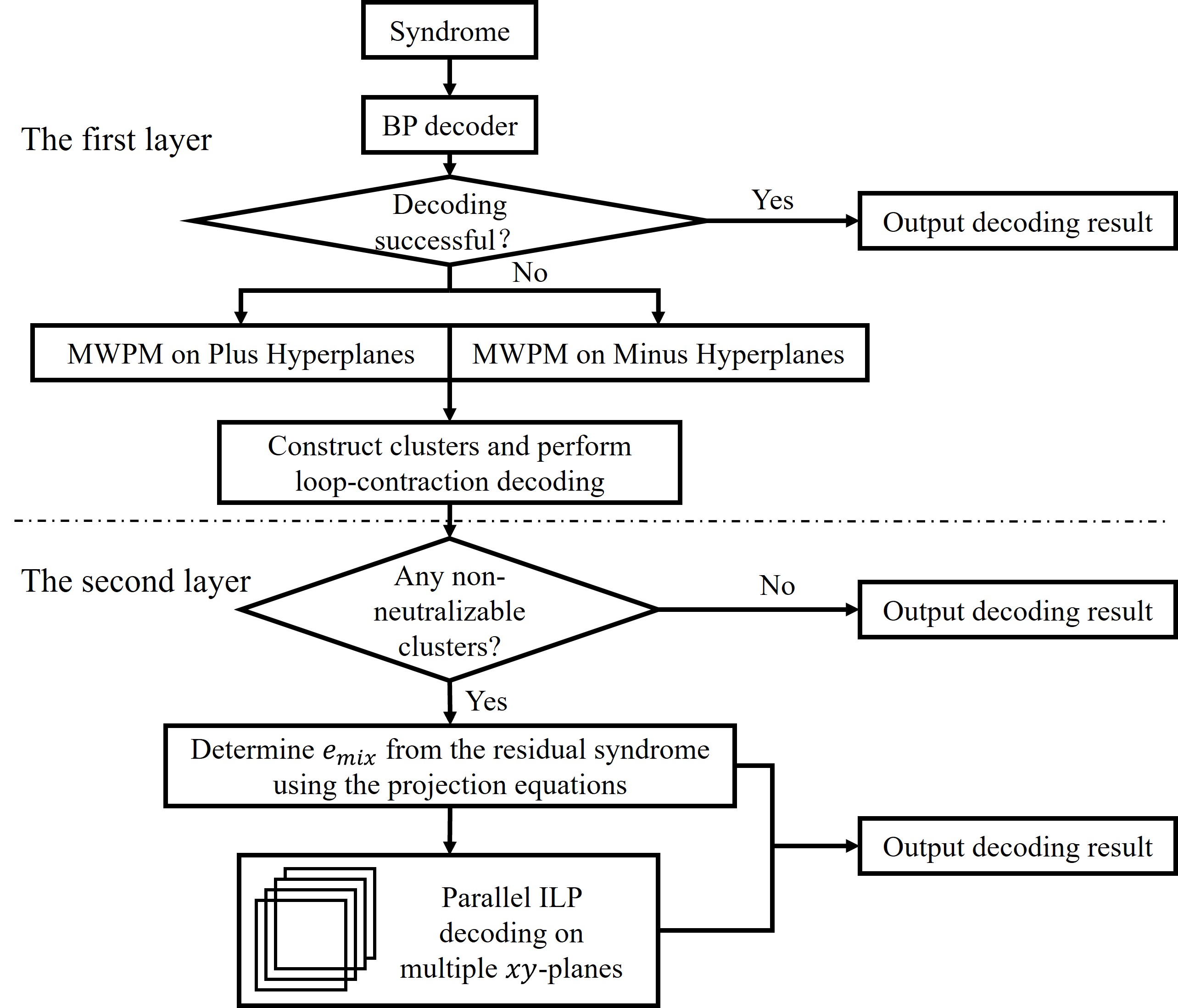}
	\caption{The flowchart of double-layer decoding strategy.}
	\label{DD_decoder_flowchart}
\end{figure}

\subsection{Decoding with symmetries}
\label{decoding with symmetry}
As discussed in Sec.~\ref{Quadrupole}, all physical qubits lie on hyperplanes whose normal vectors $(r_x,r_y,r_z,r_w)$ satisfy
$r_x,r_y,r_z,r_w\in\{\pm1\}$ and constant terms $E$ are even, whereas all stabilizers lie on hyperplanes with odd values of $E$. Moreover, for an isotropic lattice $L_x=L_y=L_z=L_w$, any qubit lying on a hyperplane with even $E$ is incident to four $X$-type or four $Z$-type stabilizers, two of which lie on the hyperplane with constant term $(E-1)\bmod g$, while the other two lie on the hyperplane with constant term $(E+1)\bmod g$. Consequently, the defects on each hyperplane can be paired using the conventional MWPM algorithm on an ordinary graph.

Although performing MWPM independently on each hyperplane can yield candidate erroneous qubits from the resulting matchings, these candidates do not, in general, constitute a valid decoding solution. This is because such a procedure is equivalent to decomposing the linear system in Eq.~\eqref{IS} as
\begin{equation}
	\label{IS_trans}
	H\hat{\mathbf{e}}=\mathbf{s}
	\quad\Longrightarrow\quad
	\begin{bmatrix}
		H_1\mathbf{e}_1\\
		H_2\mathbf{e}_2\\
		\vdots\\
		H_L\mathbf{e}_L
	\end{bmatrix}
	=
	\begin{bmatrix}
		\mathbf{s}_1\\
		\mathbf{s}_2\\
		\vdots\\
		\mathbf{s}_L
	\end{bmatrix},
\end{equation}
and then solving the subproblems
$H_i\mathbf{e}_i=\mathbf{s}_i$ independently for $i=1,\ldots,L$, while simultaneously requiring the consistency condition $\mathbf{e}_1=\mathbf{e}_2=\cdots=\mathbf{e}_L$. However, this global consistency condition cannot, in general, be guaranteed because the decoding problems on different hyperplanes are solved independently. Even exploiting the correlating matching\cite{} cannot guarantee this global consistency condition to be satisfied. Therefore, at this stage, we retain only the matching results obtained by MWPM on each hyperplane, rather than directly using the corresponding inferred erroneous qubits.

The objective of the first layer of decoding is to decompose the original decoding problem into multiple smaller and tractable subproblems while ensuring that the resulting solutions remain valid for the original decoding problem. In this paper, we achieve this by performing MWPM separately on the two families of hyperplanes with normal vectors $(1,1,1,1)$ and $(1,-1,1,-1)$, and then constructing neutralizable clusters from the corresponding matching results. A key task at this stage is to determine whether a given cluster can be neutralized. A neutralizable cluster is defined in Definition~\ref{definition Neutralizable cluster}.

For convenience, we refer to the set of all hyperplanes with normal vector
$(1,1,1,1)$ as the plus hyperplane family, and the set of all hyperplanes
with normal vector $(1,-1,1,-1)$ as the minus hyperplane family.
The normal vectors of the two hyperplane families are orthogonal, since $(1,1,1,1)\cdot(1,-1,1,-1)=0$.

\begin{definition}[\textbf{Neutralizable cluster}]
	\label{definition Neutralizable cluster}
	Let $D$ denote the set of all defects. MWPM is performed independently
	on each hyperplane in the plus and minus hyperplane families, yielding
	two sets of matching edges, denoted by $\mathcal{M}_{+}$ and
	$\mathcal{M}_{-}$, respectively. From these matching results, we construct an auxiliary graph $G_{\mathrm{match}}=\left(D,\mathcal{M}_{+}\cup\mathcal{M}_{-}\right)$, whose vertex set is $D$ and whose edge set is
	$\mathcal{M}_{+}\cup\mathcal{M}_{-}$.
	
	A \textbf{cluster} is defined as a connected component of
	$G_{\mathrm{match}}$. A cluster is \textbf{neutralizable} if there exists a Pauli operator whose syndrome is exactly the set of
	defects contained in that cluster.
\end{definition}

The construction of clusters is the key step that enables the original
decoding problem to be decomposed into multiple smaller subproblems using
the two sets of matching results. Determining whether each resulting
subproblem admits a valid solution is therefore equivalent to determining
whether the corresponding cluster is neutralizable. In the following,
we derive an algebraic criterion for testing the neutralizability of a
cluster.

\begin{corollary}[\textbf{Neutralizability criterion for a cluster}]
	\label{Neutralizable cluster}
	Let $C^X$ and $C^Z$ be clusters consisting of $X$-type and $Z$-type defects, respectively, with corresponding syndrome vectors $\mathbf{s}^x$ and $\mathbf{s}^z$. If $C^X$ and $C^Z$ are neutralizable, then their corresponding syndrome vectors necessarily satisfy
	\begin{equation}
		\left[
		N_1\otimes N_2[0:m_4],\,
		N_1\otimes N_2[m_4:m_4+m_3]
		\right]\mathbf{s}^x=\mathbf{0}
	\end{equation}
	and
	\begin{equation}
		\left[
		N_4[0:m_1]\otimes N_3,\,
		N_4[m_1:m_1+m_2]\otimes N_3
		\right]\mathbf{s}^z=\mathbf{0},
	\end{equation}
	respectively. Here, $N_1$, $N_2$, $N_3$, and $N_4$ are matrices whose rows form bases for the kernels of $\begin{bmatrix}
		H_{x_1}\\
		H_{z_1}
	\end{bmatrix}$, $[H_{x_2}^{T},H_{z_2}^{T}]$, $\begin{bmatrix}
		H_{x_2}\\
		H_{z_2}
	\end{bmatrix}$ and $[H_{x_1}^{T},H_{z_1}^{T}]$, respectively. The notation $N_i[a:b]$ denotes the submatrix of $N_i$ consisting of
	columns with indices $a,a+1,\ldots,b-1$, where the column indices start from 0.
\end{corollary}

\begin{proof}
	We first consider an $X$-type defect cluster $C^X$. If $C^X$ is
	neutralizable, then its syndrome vector $\mathbf{s}^x$ must belong
	to the image of the parity-check matrix $H_x$, i.e., $\mathbf{s}^x\in\operatorname{Im}(H_x)$. Hence, there exists an error vector $\mathbf e$ such that $\mathbf{s}^x=H_x\mathbf e$. It is therefore sufficient to show that
	\begin{equation}
		\left[
		N_1\otimes N_2[0:m_4],\,
		N_1\otimes N_2[m_4:m_4+m_3]
		\right]H_x=\mathbf{0}.
	\end{equation}
	Substituting the block-matrix expression for $H_x$ from	Eq.~\eqref{4DXYZ stabilizer_CSS}, we obtain
	\begin{widetext}
		\begin{equation}
				\left[
				N_1\otimes N_2[0:m_4],\,
				N_1\otimes N_2[m_4:m_4+m_3]
				\right]
				\begin{bmatrix}
					H_{z_1}^{T}\otimes I_{m_4}
					& \mathbf{0}
					& I_{n_A}\otimes H_{x_2}
					& H_{x_1}^{T}\otimes I_{m_4}
					& \mathbf{0}\\
					\mathbf{0}
					& H_{z_1}^{T}\otimes I_{m_3}
					& I_{n_A}\otimes H_{z_2}
					& \mathbf{0}
					& H_{x_1}^{T}\otimes I_{m_3}
				\end{bmatrix}=\mathbf{0}.
		\end{equation}
	\end{widetext}
	
	The same argument applies to a $Z$-type defect cluster $C^Z$.
	Substituting the block-matrix expression for $H_z$ from	Eq.~\eqref{4DXYZ stabilizer_CSS} gives
	\begin{widetext}
	\begin{equation}
			\left[
			N_4[0:m_1]\otimes N_3,\,
			N_4[m_1:m_1+m_2]\otimes N_3
			\right]\begin{bmatrix}
				I_{m_1}\otimes H_{x_2}^{T}
				& I_{m_1}\otimes H_{z_2}^{T}
				& H_{z_1}\otimes I_{n_B}
				& \mathbf{0}
				& \mathbf{0}\\
				\mathbf{0}
				& \mathbf{0}
				& H_{x_1}\otimes I_{n_B}
				& I_{m_2}\otimes H_{x_2}^{T}
				& I_{m_2}\otimes H_{z_2}^{T}
			\end{bmatrix}=\mathbf{0}.
	\end{equation}
	\end{widetext}
	Thus, every neutralizable $X$-type or $Z$-type cluster satisfies the
	corresponding constraint, which completes the proof.
\end{proof}

\textbf{Corollary}~\ref{Neutralizable cluster} provides an efficient criterion for testing the neutralizability of a cluster. Moreover, the matrices $N_1$, $N_2$, $N_3$, and $N_4$ can be precomputed and therefore introduce no additional computational overhead during decoding. It should be emphasized, however, that the clusters obtained from the matching procedure are not necessarily neutralizable. The reason for this will be discussed in detail in Sec.~\ref{4D Chamon code projection}. Non-neutralizable clusters are handled by the second layer of the decoder. For each neutralizable cluster, we construct a plus-minus alternating cycle and decode the cluster by progressively contracting this cycle. The formal definition of a plus-minus alternating cycle is given below.

\begin{definition}[\textbf{Plus-minus alternating cycle}]
	Let $C$ be a neutralizable cluster constructed from plus-matching and
	minus-matching edges. Denote the set of defects in $C$ by $V_C$, the set
	of plus-matching edges by $\mathcal{M}^{+}_C$, and the set of
	minus-matching edges by $\mathcal{M}^{-}_C$. The matching graph associated
	with $C$ is defined as $G_C=
		\left(
		V_C,\,
		\mathcal{M}^{+}_C\cup\mathcal{M}^{-}_C
		\right)$. A simple cycle $\gamma=
		\left(
		v_0,v_1,\ldots,v_{2m-1},v_0
		\right)$
	in $G_C$ is called a \emph{plus-minus alternating cycle} if its edges $e_i=(v_i,v_{i+1}),\ i=0,1,\ldots,2m-1$, where the indices are taken modulo $2m$, strictly alternate between plus-matching and minus-matching edges. Equivalently,
	\begin{equation}
		e_i\in\mathcal{M}^{+}_C
		\Longleftrightarrow
		e_{i+1}\in\mathcal{M}^{-}_C,\ i\bmod 2m.
	\end{equation}
\end{definition}

From the chain-complex representation of CSS stabilizer codes, a syndrome
generated by a set of physical-qubit errors can be regarded as the image of
the corresponding error support under the appropriate boundary map. Hence,
for a neutralizable $X$-type or $Z$-type cluster, there exists a set of
physical-qubit errors whose syndrome is exactly the set of defects contained
in that cluster. In the matching representation introduced above, such a
neutralizable cluster can be organized into a closed plus-minus alternating
cycle connecting its defects. Accordingly, the first-layer decoding of a
neutralizable cluster can be interpreted as progressively contracting the
corresponding plus-minus alternating cycle until all defects are annihilated.

Fig. \ref{plus_minus_cycle}(a) and (b) show an example of the plus-matching and minus-matching results of the 4D Chamon code on a 2D $xy$-plane, respectively. Fig. \ref{plus_minus_cycle}(c) illustrates how a plus-minus alternating cycle of a neutralizable cluster is constructed from the two matching results. The physical qubits enclosed by this cycle correspond to the erroneous qubits.

\begin{figure*}[htbp]
	\centering
	\includegraphics[width=0.8\textwidth]{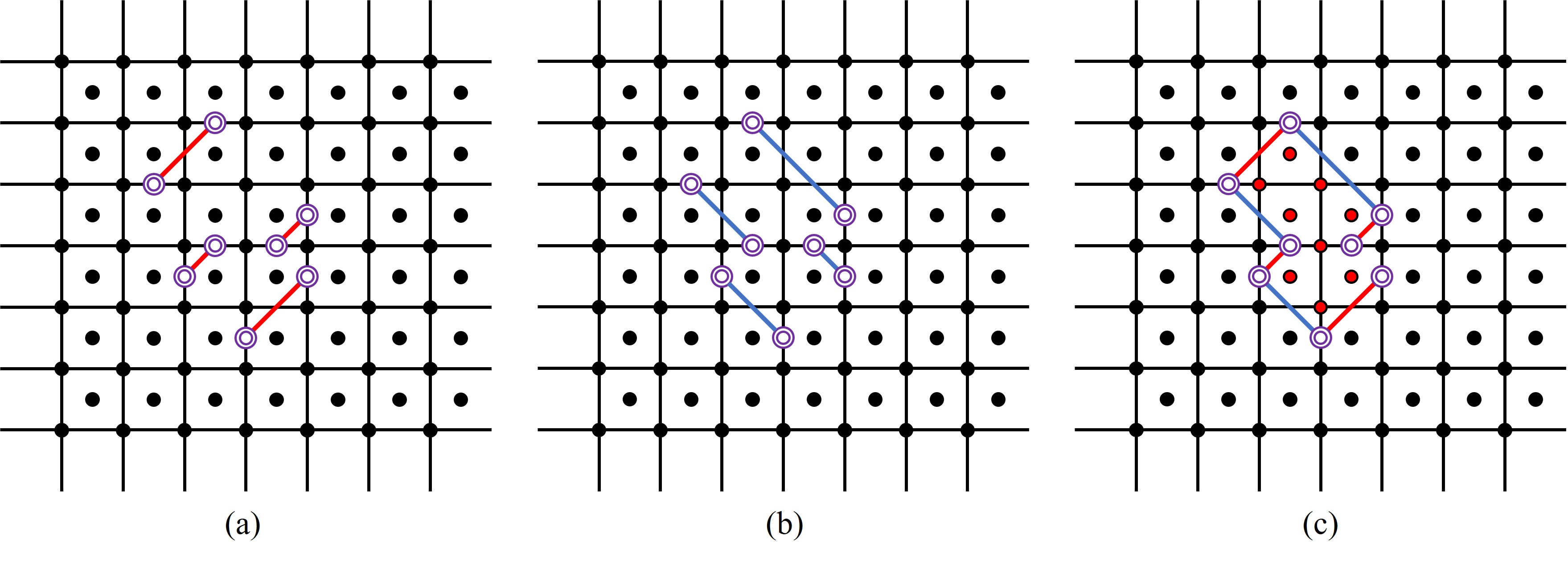}
	\caption{Constructing plus-minus alternating cycle on a 2D $xy$-plane according to the matching results of plus and minus hyperplanes.}
	\label{plus_minus_cycle}
\end{figure*}

\subsection{Decoding with projection}
\label{4D Chamon code projection}
This section demonstrates that the 4D Chamon code exhibits a projection-induced 2D toric-code structure from the perspective of its linear-system formulation. As can be seen from Eq.~\eqref{4DXYZ stabilizer_CSS}, the parity-check matrices $H_x$ and $H_z$ have analogous algebraic structures. Therefore, we consider only the correction of Pauli $Z$ errors in the following.

For Pauli $Z$ errors, the decoding problem is equivalent to solving the linear system $H_x\mathbf{e}_z=\mathbf{s}\pmod 2$. According to the geometric distribution of the physical qubits, the error vector $\mathbf{e}_z$ can be decomposed as
\begin{equation}
	\mathbf{e}_z
	=
	\left(
	\mathbf{e}_{xy},
	\mathbf{e}_{hc},
	\mathbf{e}_{mix},
	\mathbf{e}_{v},
	\mathbf{e}_{zw}
	\right)^T,
\end{equation}
where $\mathbf{e}_{xy}$, $\mathbf{e}_{hc}$, $\mathbf{e}_{v}$, and $\mathbf{e}_{zw}$ denote the error vectors associated with qubits located at the centers of $xy$-faces, at the centers of hypercubes, at vertices, and at the centers of $zw$-faces, respectively. The vector $\mathbf{e}_{mix}$ represents errors on qubits located at the centers of $xz$-, $xw$-, $yz$-, and $yw$-faces. In the following, these four types of faces are collectively referred to as mixed faces.

In addition, the syndrome vector $\mathbf{s}$ can be decomposed into two components, denoted by $\mathbf{s}_{1}$ and $\mathbf{s}_{2}$. Using the block-matrix expression for $H_x$ in Eq.~\eqref{4DXYZ stabilizer_CSS}, the decoding equation can be written as
\begin{widetext}
	\begin{equation}
		\label{decoding Hx 1}
		\begin{bmatrix}
			H_{z_1}^T\otimes I_{m_4} & \mathbf{0} & I_{n_A}\otimes H_{x_2} & H_{x_1}^T\otimes I_{m_4} & \mathbf{0}\\
			\mathbf{0} & H_{z_1}^T\otimes I_{m_3} & I_{n_A}\otimes H_{z_2} & \mathbf{0} & H_{x_1}^T\otimes I_{m_3}
		\end{bmatrix}
		\begin{bmatrix}
			\mathbf{e}_{xy}\\
			\mathbf{e}_{hc}\\
			\mathbf{e}_{mix}\\
			\mathbf{e}_{v}\\
			\mathbf{e}_{zw}
		\end{bmatrix}
		=
		\begin{bmatrix}
			\mathbf{s}_{1}\\
			\mathbf{s}_{2}
		\end{bmatrix}.
	\end{equation}
\end{widetext}

As can be seen from Eq.~\eqref{decoding Hx 1}, the error vector $\mathbf{e}_{mix}$ contributes to both syndrome components $\mathbf{s}_{1}$ and $\mathbf{s}_{2}$ and therefore couples the two subproblems. If errors on the mixed faces are temporarily neglected, i.e., $\mathbf{e}_{mix}=\mathbf{0}$, Eq.~\eqref{decoding Hx 1} reduces to
\begin{widetext}
	\begin{equation}
		\label{decoding Hx 2}
		\begin{bmatrix}
			H_{z_1}^T\otimes I_{m_4} & \mathbf{0} & H_{x_1}^T\otimes I_{m_4} & \mathbf{0}\\
			\mathbf{0} & H_{z_1}^T\otimes I_{m_3} & \mathbf{0} & H_{x_1}^T\otimes I_{m_3}
		\end{bmatrix}
		\begin{bmatrix}
			\mathbf{e}_{xy}\\
			\mathbf{e}_{hc}\\
			\mathbf{e}_{v}\\
			\mathbf{e}_{zw}
		\end{bmatrix}
		=
		\begin{bmatrix}
			\mathbf{s}_{1}\\
			\mathbf{s}_{2}
		\end{bmatrix}.
	\end{equation}
\end{widetext}

Equation~\eqref{decoding Hx 2} can then be decomposed into the following two independent subproblems:
\begin{equation}
	\label{decoding Hx 3}
	\left[
	H_{z_1}^T\otimes I_{m_4},\,
	H_{x_1}^T\otimes I_{m_4}
	\right]
	\begin{bmatrix}
		\mathbf{e}_{xy}\\
		\mathbf{e}_{v}
	\end{bmatrix}
	=
	\mathbf{s}_{1},
\end{equation}
and
\begin{equation}
	\label{decoding Hx 4}
	\left[
	H_{z_1}^T\otimes I_{m_3},\,
	H_{x_1}^T\otimes I_{m_3}
	\right]
	\begin{bmatrix}
		\mathbf{e}_{hc}\\
		\mathbf{e}_{zw}
	\end{bmatrix}
	=
	\mathbf{s}_{2}.
\end{equation}

This indicates that, in the absence of mixed-face errors, the original
decoding problem in four dimensions can be decomposed into two families of
2D projected subproblems. Specifically, the matrix $\left[
H_{z_1}^{T}\otimes I_{m_4},\,
H_{x_1}^{T}\otimes I_{m_4}
\right]$ in Eq.~\eqref{decoding Hx 3} admits the following physical interpretation: the roles of stabilizers and physical qubits in the original 2D toric code $C_1$ are interchanged, and the resulting structure is replicated $m_4$ times. In the original toric code, edges correspond to physical qubits, whereas vertices and faces correspond to stabilizers. After this interchange, edges correspond to stabilizers, whereas vertices and faces correspond to physical qubits. The resulting structure can therefore be regarded as a collection of independent $xy$-planes, denoted by
$\Pi_{xy}(z_0,w_0)$. Within each $xy$-plane with fixed $(z_0,w_0)$,
errors on physical qubits flip the measurement outcomes of four stabilizers
within the same plane, while different $xy$-planes are
decoupled from one another. Similarly, the matrix $\left[
H_{z_1}^{T}\otimes I_{m_3},\,
H_{x_1}^{T}\otimes I_{m_3}
\right]$ in Eq.~\eqref{decoding Hx 4} corresponds to another collection of independent $xy$-planes. Therefore, when $\mathbf{e}_{mix}=\mathbf{0}$, the decoding problem for Pauli $Z$ errors can be decomposed into a collection of independent 2D toric-code-like subproblems defined on $xy$-planes.

On the other hand, if only mixed-plane errors are considered, i.e.,
$\mathbf{e}_{xy}
=
\mathbf{e}_{hc}
=
\mathbf{e}_{v}
=
\mathbf{e}_{zw}
=
\mathbf{0}$, then Eq.~\eqref{decoding Hx 1} reduces to
\begin{equation}
	\label{decoding Hx 5}
	\begin{bmatrix}
		I_{n_A}\otimes H_{x_2}\\
		I_{n_A}\otimes H_{z_2}
	\end{bmatrix}
	\mathbf{e}_{mix}
	=
	\begin{bmatrix}
		\mathbf{s}_1\\
		\mathbf{s}_2
	\end{bmatrix}.
\end{equation}
In this case, the syndrome generated by $\mathbf{e}_{mix}$ can be decomposed
into a collection of $zw$-planes $\Pi_{zw}(x_0,y_0)$ with fixed
$(x_0,y_0)$. More specifically, the matrix $\begin{bmatrix}
	I_{n_A}\otimes H_{x_2}\\
	I_{n_A}\otimes H_{z_2}
\end{bmatrix}$ can be interpreted as $n_A$ independent copies of the original toric code $C_2$, giving rise to a collection of independent 2D $zw$-planes $\Pi_{zw}(x_0,y_0)$. Errors on physical qubits within a given plane flip the measurement outcomes of four stabilizers within the same plane, while different $zw$-planes are decoupled from one another. Thus, mixed-face errors also give rise to a collection of
independent 2D toric-code-like subproblems, except that the projection is
onto $zw$-planes rather than $xy$-planes.

The decoding problem for Pauli $X$ errors can be analyzed in exactly the same manner. In this case, the linear system $H_z\mathbf{e}_x=\mathbf{s}\pmod 2$ equivalent to the decoding problem can be also written as
\begin{widetext}
	\begin{equation}
		\label{decoding Hz 1}
		\begin{bmatrix}
			I_{m_1}\otimes H_{x_2}^{T}
			&
			I_{m_1}\otimes H_{z_2}^{T}
			&
			H_{z_1}\otimes I_{n_B}
			&
			\mathbf{0}
			&
			\mathbf{0}
			\\
			\mathbf{0}
			&
			\mathbf{0}
			&
			H_{x_1}\otimes I_{n_B}
			&
			I_{m_2}\otimes H_{x_2}^{T}
			&
			I_{m_2}\otimes H_{z_2}^{T}
		\end{bmatrix}
		\begin{bmatrix}
			\mathbf{e}_{xy}\\
			\mathbf{e}_{hc}\\
			\mathbf{e}_{mix}\\
			\mathbf{e}_{v}\\
			\mathbf{e}_{zw}
		\end{bmatrix}
		=
		\begin{bmatrix}
			\mathbf{s}_1\\
			\mathbf{s}_2
		\end{bmatrix}.
	\end{equation}
\end{widetext}
Similarly, let $\mathbf{e}_{mix}=\mathbf{0}$, Eq.~(\ref{decoding Hz 1}) reduces to
\begin{widetext}
	\begin{equation}
		\label{decoding Hz 2}
		\begin{bmatrix}
			I_{m_1}\otimes H_{x_2}^{T}
			&
			I_{m_1}\otimes H_{z_2}^{T}
			&
			\mathbf{0}
			&
			\mathbf{0}
			\\
			\mathbf{0}
			&
			\mathbf{0}
			&
			I_{m_2}\otimes H_{x_2}^{T}
			&
			I_{m_2}\otimes H_{z_2}^{T}
		\end{bmatrix}
		\begin{bmatrix}
			\mathbf{e}_{xy}\\
			\mathbf{e}_{hc}\\
			\mathbf{e}_{v}\\
			\mathbf{e}_{zw}
		\end{bmatrix}
		=
		\begin{bmatrix}
			\mathbf{s}_1\\
			\mathbf{s}_2
		\end{bmatrix}.
	\end{equation}
\end{widetext}
The two subequations in Eq.~(\ref{decoding Hz 2}) correspond to collections of mutually independent decoding problems on $zw$-planes. Specifically, the matrix $\left[
I_{m_1}\otimes H_{x_2}^{T},\ 
I_{m_1}\otimes H_{z_2}^{T}
\right]$
can be interpreted as the structure obtained by interchanging the roles of stabilizers and physical qubits in the original 2D toric code $C_2$ and then replicating the resulting structure $m_1$ times. It therefore corresponds to a collection of mutually independent $zw$-planes $\Pi_{zw}(x_0,y_0)$. Likewise, the matrix $\left[
I_{m_2}\otimes H_{x_2}^{T},\ 
I_{m_2}\otimes H_{z_2}^{T}
\right]$ corresponds to another collection of mutually independent $zw$-planes. Therefore, in the absence of mixed-face errors, the decoding problem for Pauli $X$ errors can likewise be decomposed into multiple mutually independent 2D toric-code-like subproblems.

Furthermore, if only mixed-face errors are considered, i.e., $\textbf{e}_{xy}=\textbf{e}_{hc}=\textbf{e}_{v}=\textbf{e}_{zw}=\textbf{0}$, then Eq.~(\ref{decoding Hz 1}) reduces to
\begin{equation}
	\label{decoding Hz 3}
	\begin{bmatrix}
		H_{z_1}\otimes I_{n_B}\\
		H_{x_1}\otimes I_{n_B}
	\end{bmatrix}
	\mathbf{e}_{mix}
	=
	\begin{bmatrix}
		\mathbf{s}_1\\
		\mathbf{s}_2
	\end{bmatrix}.
\end{equation}
This equation shows that Pauli $X$ errors on the mixed planes can likewise be projected onto a collection of independent $xy$-plane subproblems. In particular, the matrix
$\begin{bmatrix}
	H_{z_1}\otimes I_{n_B}\\
	H_{x_1}\otimes I_{n_B}
\end{bmatrix}$ can be interpreted physically as $n_B$ independent copies of the original toric code $C_1$, giving rise to a collection of mutually independent 2D $xy$-planes $\Pi_{xy}(z_0,w_0)$. Errors on physical qubits within a given plane flip four stabilizers within the same plane, while different $xy$-planes do not affect one another.

In summary, for the 4D Chamon code, the syndromes generated by both Pauli $Z$ and Pauli $X$ errors exhibit a projection-induced 2D toric-code structure. Specifically, once the relevant error components are decoupled, the original decoding problem can be decomposed into a collection of mutually independent 2D toric-code-like subproblems defined on either $xy$- or $zw$-planes.

In the general case, however, the syndrome contributions generated by mixed-face errors overlap with those generated by errors at the remaining qubit locations, thereby coupling the projected decoding problems associated with different 2D planes. Consequently, the origin of a given defect may become ambiguous: it may be associated with either an $xy$-plane subproblem or a $zw$-plane subproblem, particularly at relatively high physical error rates. This ambiguity is also one of the main reasons why the construction of plus-minus alternating cycles described in Sec.~\ref{decoding with symmetry} may occasionally fail.

We next establish the central theoretical result underlying the second layer of the decoder. By first solving for the mixed-face error component $\mathbf{e}_{mix}$, the contribution of mixed-face errors can be separated from that of errors at the remaining qubit locations. Consequently, the original 4D decoding problem can be decomposed into a collection of mutually independent 2D toric-code-like subproblems defined on $xy$-planes.

\begin{proposition}[\textbf{Equivalence between the solvability of the projected and original decoding equations}]
	\label{projection equivalence}
	Since the $X$-type and $Z$-type parity-check matrices $H_x$ and $H_z$
	of the 4D Chamon code have analogous algebraic structures, we restrict
	our discussion to $H_x$ to avoid repetition. All linear systems below
	are considered over $\mathbb{F}_2$. Equation~\eqref{decoding Hx 1}
	can be rewritten as
	\begin{equation}
		\label{primal equation}
		A\mathbf{e}_0+C\mathbf{e}_{mix}
		=
		\begin{bmatrix}
			\mathbf{s}_1\\
			\mathbf{s}_2
		\end{bmatrix},
	\end{equation}
	where
	\begin{equation}
		\mathbf{e}_0=
		\left(
		\mathbf{e}_{xy},
		\mathbf{e}_{hc},
		\mathbf{e}_{v},
		\mathbf{e}_{zw}
		\right)^T,
	\end{equation}
	and
	\begin{equation}
		A=
		\begin{bmatrix}
			H_{z_1}^{T}\otimes I_{m_4}
			& \mathbf{0}
			& H_{x_1}^{T}\otimes I_{m_4}
			& \mathbf{0}
			\\
			\mathbf{0}
			& H_{z_1}^{T}\otimes I_{m_3}
			& \mathbf{0}
			& H_{x_1}^{T}\otimes I_{m_3}
		\end{bmatrix},
	\end{equation}
	\begin{equation}
		C=
		\begin{bmatrix}
			I_{n_A}\otimes H_{x_2}\\
			I_{n_A}\otimes H_{z_2}
		\end{bmatrix}.
	\end{equation}
	
	Let
	\begin{equation}
		M=
		\begin{bmatrix}
			H_{x_1}\\
			H_{z_1}
		\end{bmatrix},
	\end{equation}
	and let the rows of $N_1$ form a basis for $\ker M$. Define
	\begin{equation}
		P=
		\begin{bmatrix}
			N_1\otimes I_{m_4}
			& \mathbf{0}
			\\
			\mathbf{0}
			& N_1\otimes I_{m_3}
		\end{bmatrix}.
	\end{equation}
	Then Eq.~\eqref{primal equation} is solvable if and only if the
	projected equation
	\begin{equation}
		PC\mathbf{e}_{mix}
		=
		P
		\begin{bmatrix}
			\mathbf{s}_1\\
			\mathbf{s}_2
		\end{bmatrix}
	\end{equation}
	is solvable. Equivalently,
	\begin{equation}
		\label{projection equation}
		\begin{bmatrix}
			N_1\otimes H_{x_2}\\
			N_1\otimes H_{z_2}
		\end{bmatrix}
		\mathbf{e}_{mix}
		=
		\begin{bmatrix}
			N_1\otimes I_{m_4}
			& \mathbf{0}
			\\
			\mathbf{0}
			& N_1\otimes I_{m_3}
		\end{bmatrix}
		\begin{bmatrix}
			\mathbf{s}_1\\
			\mathbf{s}_2
		\end{bmatrix}
	\end{equation}
	has a solution if and only if Eq.~\eqref{primal equation} has a solution.
\end{proposition}

\begin{proof}
	Since the rows of $N_1$ form a basis for $\ker
	\begin{bmatrix}
		H_{x_1}\\
		H_{z_1}
	\end{bmatrix}$,
	we have $N_1H_{x_1}^{T}=\mathbf{0}$ and $N_1H_{z_1}^{T}=\mathbf{0}$. It follows immediately that $PA=\mathbf{0}$.
	
	Suppose first that Eq.~\eqref{primal equation} is solvable.
	Left-multiplying both sides by $P$ gives
	\begin{equation}
		PA\mathbf{e}_0+PC\mathbf{e}_{mix}
		=
		P
		\begin{bmatrix}
			\mathbf{s}_1\\
			\mathbf{s}_2
		\end{bmatrix}.
	\end{equation}
	Since $PA=\mathbf{0}$, it follows that
	\begin{equation}
		PC\mathbf{e}_{mix}
		=
		P
		\begin{bmatrix}
			\mathbf{s}_1\\
			\mathbf{s}_2
		\end{bmatrix},
	\end{equation}
	which proves that the projected equation is solvable.
	
	Conversely, suppose that the projected equation admits a solution
	$\mathbf{e}_{mix}$. Fix such a solution and define the residual syndrome
	by
	\begin{equation}
		\mathbf{s}^{\prime\prime}
		=
		\mathbf{s}+C\mathbf{e}_{mix},
	\end{equation}
	where $\mathbf{s}
		=
		\begin{bmatrix}
			\mathbf{s}_1\\
			\mathbf{s}_2
		\end{bmatrix}$.
	Using the projected equation, we obtain
	\begin{equation}
		P\mathbf{s}^{\prime\prime}
		=
		P\left(\mathbf{s}+C\mathbf{e}_{mix}\right)
		=
		\mathbf{0}.
	\end{equation}
	Hence, $\mathbf{s}^{\prime\prime}\in\ker P$.
	
	It remains to show that $\ker P=\operatorname{Im}A$. By construction, $\operatorname{row}(N_1)=\ker M$. Therefore, we have
	\begin{equation}
	\begin{aligned}
		\ker N_1
		&=
		\operatorname{row}(N_1)^{\perp} \\
		&=
		(\ker M)^{\perp} \\
		&=
		\operatorname{row}(M) \\
		&=
		\operatorname{Im}(M^T) \\
		&=
		\operatorname{Im}
		\left[
		H_{x_1}^{T},
		H_{z_1}^{T}
		\right].
	\end{aligned}
\end{equation}
	Since the order of the two column blocks does not affect their image,
	\begin{equation}
		\ker N_1
		=
		\operatorname{Im}
		\left[
		H_{z_1}^{T},
		H_{x_1}^{T}
		\right].
	\end{equation}
	It follows that
	\begin{equation}
		\ker
		\left(
		N_1\otimes I_{m_4}
		\right)
		=
		\operatorname{Im}
		\left[
		H_{z_1}^{T}\otimes I_{m_4},
		H_{x_1}^{T}\otimes I_{m_4}
		\right],
	\end{equation}
	and
	\begin{equation}
		\ker
		\left(
		N_1\otimes I_{m_3}
		\right)
		=
		\operatorname{Im}
		\left[
		H_{z_1}^{T}\otimes I_{m_3},
		H_{x_1}^{T}\otimes I_{m_3}
		\right].
	\end{equation}
	Therefore, from the block-diagonal structure of $P$ and the corresponding
	block structure of $A$, we obtain
	\begin{equation}
		\ker P=\operatorname{Im}A.
	\end{equation}
	
	Since $\mathbf{s}^{\prime\prime}\in\ker P$, it follows that
	$\mathbf{s}^{\prime\prime}\in\operatorname{Im}A$. Hence, there exists
	an error vector
	\begin{equation}
		\mathbf{e}_0=
		\left(
		\mathbf{e}_{xy},
		\mathbf{e}_{hc},
		\mathbf{e}_{v},
		\mathbf{e}_{zw}
		\right)^T
	\end{equation}
	such that
	\begin{equation}
		A\mathbf{e}_0=\mathbf{s}^{\prime\prime}.
	\end{equation}
	By the definition of $\mathbf{s}^{\prime\prime}$ and the fact that the
	calculations are performed over $\mathbb{F}_2$, we obtain
	\begin{equation}
		A\mathbf{e}_0+C\mathbf{e}_{mix}
		=
		\begin{bmatrix}
			\mathbf{s}_1\\
			\mathbf{s}_2
		\end{bmatrix}.
	\end{equation}
	Thus, Eq.~\eqref{primal equation} is solvable, which completes the proof.
\end{proof}

\textbf{Proposition}~\ref{projection equivalence} shows that the mixed-face error component $\mathbf{e}_{mix}$ can first be determined independently by
solving the projected equation. Once $\mathbf{e}_{mix}$ is determined,
the residual syndrome necessarily lies in $\operatorname{Im}(A)$.
Consequently, the remaining error components $\mathbf{e}_{xy}$, $\mathbf{e}_{hc}$, $\mathbf{e}_{v}$, and $\mathbf{e}_{zw}$ can be recovered independently and in parallel on a collection of mutually independent $xy$-planes. This constitutes the central idea of the second decoding layer, namely, \textbf{decoupling the original decoding problem by first solving for $\mathbf{e}_{mix}$}.

Obtaining a particular solution to the projected equation in
Eq.~\eqref{projection equation} is straightforward. For example, Gaussian
elimination can be used to directly obtain a particular solution
$\hat{\mathbf{e}}_{mix}$ satisfying the projected equation. However, such a
solution is generally not the most likely error estimate from a probabilistic
perspective. Gaussian elimination accounts only for the algebraic
solvability of the linear system and does not incorporate information about
the physical noise model or the posterior error probabilities of individual
qubits. Consequently, the resulting solution may have a relatively large
Hamming weight or may assign errors to qubits that are unlikely to be
erroneous, potentially leading to a higher logical error rate.

To address this issue, we incorporate an ordered-statistics-decoding
(OSD)-based post-processing strategy when solving the projected equation.
Specifically, the original syndrome is first passed to the BP decoder to
obtain the a posteriori log-likelihood ratios (LLRs) of all physical qubits.
The LLRs associated with the mixed-face qubits are then extracted and used
to rank the columns of the projected matrix according to their reliabilities.
This reliability ordering allows columns associated with less reliable, and
hence more plausible, error locations to be preferentially selected.

According to \textbf{Proposition}~\ref{projection equivalence}, define
\begin{equation}
	P_{mix}
	=
	\begin{bmatrix}
		N_1\otimes H_{x_2}\\
		N_1\otimes H_{z_2}
	\end{bmatrix}.
\end{equation}
Suppose that $\operatorname{rank}(P_{mix})=r$. According to the reliability
ordering described above, the first $r$ linearly independent columns are
selected from $P_{mix}$ to form a column basis. Since these columns span
$\operatorname{Im}(P_{mix})$ and the projected syndrome is guaranteed to
belong to $\operatorname{Im}(P_{mix})$, a particular solution
$\hat{\mathbf{e}}_{mix}$ can be obtained using this column basis. Compared
with direct Gaussian elimination without reliability information, this
strategy preferentially constructs solutions using qubits identified by BP
as less reliable, thereby providing a more plausible estimate of the
mixed-face error component. The corresponding decoding procedure is given in Algorithm \ref{alg:osd_projection}.

It should be emphasized that, although the proposed method adopts the basic
idea of OSD \cite{Fossorier1995OSD} post-processing, OSD is not applied to the full parity-check
matrix $H_x$. Instead, the reliability-based column ordering and basis
selection are performed only on the substantially smaller projected matrix
$P_{mix}$. Consequently, the computational cost is significantly lower than
that of directly applying OSD to the full matrix $H_x$.

Once $\hat{\mathbf{e}}_{mix}$ has been obtained, the syndrome contribution
of the mixed-face errors can be removed from the original syndrome to form
the residual syndrome. According to \textbf{Proposition}~\ref{projection equivalence}, the residual syndrome necessarily lies in the image of the matrix associated with the remaining error components. Therefore, the remaining decoding problem can be decomposed into a collection of mutually independent and parallelizable toric-code-like decoding problems defined on 2D $xy$-planes. 

\begin{algorithm}[htbp]
	\caption{Sovling projection equation based on OSD strategy}
	\label{alg:osd_projection}
	\LinesNumbered
	\KwIn{Original syndrome vector $\textbf{s}=(\textbf{s}_1,\textbf{s}_2)^T$;
	$X$-type parity-check matrix $H_x$;
	Projection matrix $P$;
	Mixed-face projection matrix $P_{mix}$ and its rank $r$}
	\KwOut{Mixed-face error component $\hat{\textbf{e}}_{mix}$}
	
	LLR = BP(\textbf{s});\\
	Computing projected syndrome vector $\textbf{s}_{proj} = P\textbf{s}$;\\
	Selecting the indexes $\mathcal{I}$ of the first $r$ linearly independent columns from $P_{mix}$ according to LLR;\\
	Constructing matrix $U=P_{mix}[\mathcal{I}]$;\\
	Computing $\hat{\textbf{e}}_{mix}^{\mathcal{I}}=U^{-1}\textbf{s}_{proj}$;\\
	$\hat{\textbf{e}}_{mix} = \left(\textbf{s}_{proj},\textbf{0}\right)$
\end{algorithm}

\subsection{Numerical simulation}
\label{simulation result}
This section presents the error-correction performance of the two-layer decoder for isotropic 4D Chamon codes. Since we consider the CSS variant of the 4D Chamon code in this paper, its $X$-type and $Z$-type stabilizer structures are dual to each other, and the corresponding decoding procedures can be carried out independently. Therefore, it is sufficient to consider only one type of error, as the decoding procedure for the other type is completely analogous. Fig.~\ref{BPOSD_compare_DecoupleDecoder} compares the decoding performance of the proposed two-layer decoder with that of BP-OSD for 4D Chamon codes of the same sizes under pure Pauli $Z$ noise. The solid curves represent the logical error rates achieved by the two-layer decoder, whereas the dashed curves correspond to those obtained by BP-OSD. As can be seen, the two-layer decoder consistently achieves lower logical error rates and exhibits a code-capacity threshold of approximately $7\%$. Moreover, the performance gap between the two decoders becomes increasingly pronounced as the system size grows. It is also worth noting that the decoding performance of BP-OSD degrades with increasing system size. This indicates that, although BP-OSD is a general-purpose decoder, it is not well suited to exploiting the specific structure of the 4D Chamon code.
\begin{figure}[htbp]
	\centering
	\includegraphics[width=0.48\textwidth]{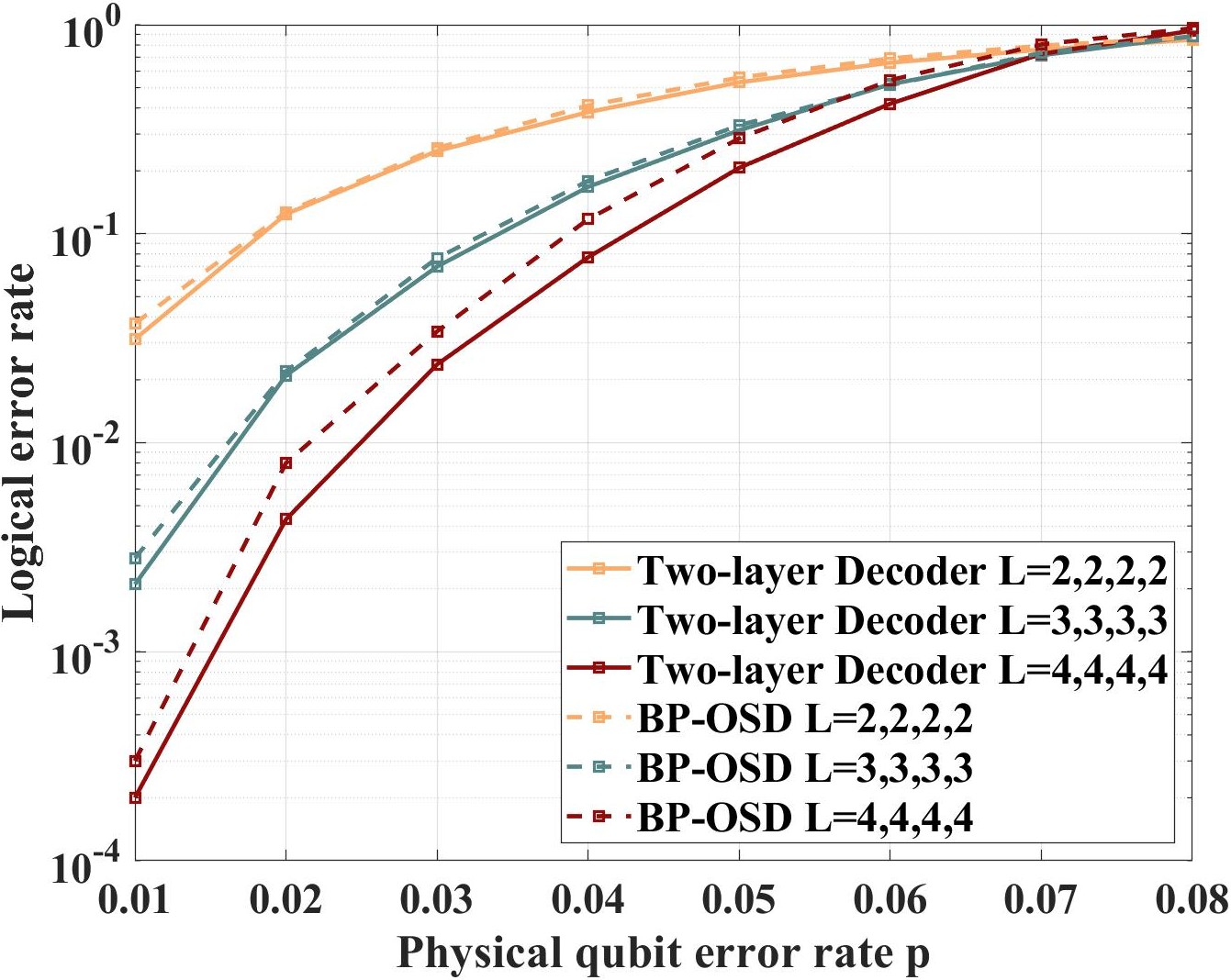}
	\caption{Comparison between the two-layer decoder and BP-OSD on the 4D Chamon code.}
	\label{BPOSD_compare_DecoupleDecoder}
\end{figure}

\section{Conclusion}
\label{conclusion}
In this paper, we study the 4D Chamon code from three perspectives: its geometric structure, excitations, and decoding strategy. First, we established the correspondence between the algebraic structure of the 4D Chamon code and the geometry of the 4D lattice, and explicitly identified the geometric distributions of physical qubits and stabilizers. Second, Building on this correspondence, we further investigated the properties and types of excitations supported by the 4D Chamon code and found that they exhibit restricted mobility, which is a key characteristic of fracton models. Moreover, the 4D Chamon code supports the same three types of excitations as the 3D Chamon code, namely, monopoles, dipoles, and quadrupoles. In both the 4D and 3D Chamon codes, monopoles and dipoles are generated by membrane operators and rigid string operators, respectively. By contrast, quadrupoles in the 4D Chamon code are generated by flexible bilayer ribbon operators, which exhibit a structure analogous to the flexible bilayer string operators that generate quadrupoles in the 3D Chamon code. These results show that the 4D Chamon code is not only a new 4D fracton model, but can also be naturally understood as a 4D generalization of the 3D Chamon code. Third, we further uncover two properties relevant to decoding: a hyperplane symmetry and a projection-induced 2D toric-code structure, which provids a new geometric perspective on the error-correction mechanism of the code. Finally, by exploiting these geometric features, we propose a two layer decoding strategy for the 4D Chamon code. Numerical simulations demonstrate that the proposed decoder achieves substantially higher decoding accuracy than the general-purpose BP-OSD decoder. Our work not only clarify the characteristics of the 4D Chamon code as a higher-dimensional fracton model, but also provide a new approach to efficient decoding for this class of higher-dimensional quantum codes.

\section*{Data Availability}
The data that support the findings of this study are available from the corresponding author upon reasonable request.

\section*{Acknowledgements}
This work is supported by the Colleges and Universities Stable Support Project of Shenzhen, China (No.GXWD20220817164856008), Shenzhen Science and Technology Program, China (JCYJ20241202123906009), Guangdong Provincial Key Laboratory of Novel Security Intelligence Technologies (2022B1212010005), the Colleges and Universities Stable Support Project of Shenzhen, China (No.GXWD20220811170225001) and Harbin Institute of Technology, Shenzhen - SpinQ quantum information Joint Research Center Project (No.HITSZ20230111).

\appendix

\newpage
\bibliography{sn-bibliography}

\end{document}